\documentclass[journal]{IEEEtran}
\usepackage{amsmath}
\usepackage{amsthm}

\usepackage{paralist}
\usepackage{multirow}
\usepackage{graphicx}
\usepackage{amsfonts}
\usepackage{amsmath,epsfig}
\usepackage{stfloats}
\usepackage{amssymb}
\usepackage{epstopdf}
\usepackage{color}
\usepackage{latexsym}
\usepackage{cite}
\usepackage{algorithm}
\usepackage{algorithmic}
\ifCLASSINFOpdf

\else

\fi

\begin{document}
\title{Green MU-MIMO/SIMO Switching for Heterogeneous Delay-aware Services with Constellation Optimization}

\author{Kunlun Wang,
        Wen Chen,~\IEEEmembership{Senior~Member,~IEEE},
         Jun Li,~\IEEEmembership{Member,~IEEE},
          %Qingqing Wu,~\IEEEmembership{Student~Member,~IEEE},
        Branka Vucetic,~\IEEEmembership{Fellow,~IEEE}
%\author{Kunlun Wang, and Wen Chen\\
%\IEEEauthorblockA{Department of Electronic Engineering, Shanghai Jiao Tong University, China\\
 %Email: \{kunlun1228; wenchen\}@sjtu.edu.cn\\}
\thanks{This work is supported by the National 973 Project \#2012CB316106,
by National 863 Project\#2015AA01A710, by SEU National Key Lab on Mobile Communications \#2013D11, by the National Natural Science Foundation of China No. 61501238, by the Jiangsu Provincial Science Foundation Project BK20150786, and by the Specially Appointed Professor Program in Jiangsu Province, 2015.}
        % <-this % stops a space
%\thanks{Manuscript received October 7, 2011; accepted November 17th, 2011. The associate editor coordinating the review of this paper
%and approving it for publication was H. Viswanathan.}
\thanks{Kunlun Wang and Wen Chen are with Shanghai Key Laboratory of Navigation and Location Based Services, Shanghai Jiao Tong University, and School of Electronic Engineering and Automation, Guilin University of Electronic Technology. (e-mails:\{kunlun1228\emph{};wenchen\}@sjtu.edu.cn).}
 \thanks{Jun~Li is with the School of Electronic and Optical Engineering, Nanjing University of Science and Technology, Nanjing, China, 210094. E-mail: jun.li@njust.edu.cn.}
\thanks{Branka Vucetic is with the School of Electrical and Information Engineering, The University of Sydney, NSW, 2006, Australia, (e-mail: branka.vucetic@sydney.edu.au).}
%\thanks{X.~Wang is with DOCOMO Beijing Communications Laboratories Co., Ltd~(e-mail:wangxl@docomolabs-beijing.com.cn).}

%\thanks{This work is supported by NSF China \#60972031, by national 973 project \#2012CB316106 and
%\#2009CB824904, by national huge special project \#2012ZX03004004,
%by national key laboratory project \#ISN11-01, by Huawei Funding
%\#YBWL2010KJ013, and by Foundation of GuangXi
%University~\#XGL090033.}
 }
%\markboth{IEEE Wireless Communications Letters ,~Vol.~1, No.~1, January~2012}%

%{Shell \MakeLowercase{\textit{et al.}}: Bare Demo of IEEEtran.cls for Journals}

\maketitle

\begin{abstract}
%Properly designed the source precoding matrix~(PM) and relay beamforming matrix~(BM) can significantly improve the spectral efficiency of the multiple-input multiple-output~(MIMO) relaying broadcast channel~(BC).
In this paper, we propose adaptive techniques for multi-user multiple input and multiple output~(MU-MIMO) cellular communication systems, to solve the problem of energy efficient communications with heterogeneous delay-aware traffic.
In order to minimize the total transmission power of the MU-MIMO, we investigate the relationship between the transmission power and the M-ary quadrature amplitude modulation~(MQAM) constellation size and get the energy efficient modulation for each transmission stream based on the minimum mean square error~(MMSE) receiver.
%The optimization algorithm is essential to coordinate the transmission of the antennas, in order to guarantee the energy-efficient transmission in particular symbol error probability.
%The users should choose the optimal constellation size in different antennas, and the algorithm can also extend to optimize the energy-efficiency with the delay constraint in the interference model.
Since the total power consumption is different for MU-MIMO and multi-user single input and multiple output~(MU-SIMO), by exploiting the intrinsic relationship among the total power consumption model, and heterogeneous delay-aware services, we propose an adaptive transmission strategy, which is a switching between MU-MIMO and MU-SIMO. Simulations show that in order to maximize the energy efficiency and consider different Quality of Service~(QoS) of delay for the users simultaneously, the users should adaptively choose the constellation size for each stream as well as the transmission mode. %In this framework, the tradeoff between energy efficiency and delay demand are well demonstrated.

\end{abstract}
\IEEEpeerreviewmaketitle

\section{Introduction}
\IEEEPARstart{R}{ecently}, energy efficient communications in wireless cellular networks have attracted
much research attention. %In communication theory, the throughput and the power are respectively the common measures of the benefit and the cost of a communication system, while the energy efficiency, expressed as the throughput per power, is to use the power as efficiently as possible.
While the battery development has not kept up with the growing demand of ubiquitous multimedia communications, the energy efficiency is more and more important for mobile users.
Meanwhile, in addition to the energy efficient wireless communications, future wireless networks are also expected to support a variety of services with delay requirements, which is one of the major Quality of Service~(QoS) for users. %this paper also addresses this critical issue.
%Thus, the throughput and the relative power consumption constitute the the energy-efficiency of the communication network.
%The throughput has been a central issue in wireless networks, the goal of the throughput is to maximize the data rate of transmission while meeting the needs of bit error rate~(BER)~\cite{1998-Goldsmith-Adaptive-modulation}. In~\cite{2009-Poor-EE}, the throughput is defined as the number of bits transmitted without error per unit time, also being called goodput, is a key measure of Quality of Service~(QoS) for wireless transmission systems~\cite{2006-Taesang-Yoo-Throughput}. And other research on the network throughput has been shown in~\cite{2000-Kumar-Capacity,2002-Tse-capacity}.
%
In~\cite{2011-Li-EE}, the authors have introduced many fundamental works and advanced techniques on energy-efficient communications.

Since multiple-input multiple-output~(MIMO) has potential to achieve high capacity, it has been a key technology for wireless systems~\cite{2002-Paulraj-MIMO}. As we know for a point-to-point system, multiple antennas can help the system to greatly reduce the transmit power. On the other hand, a multi-user multiple-input multiple-output~(MU-MIMO) system can provide a substantial gain by allowing multiple users to communicate in the same frequency and the same time slot~\cite{2009-Soysal-MIMO}. In current and emerging cellular networks, downlink and uplink transmissions can be realized with MIMO in the form of MU-MIMO, which can simultaneously benefit from multiple antennas and multi-user diversity gains~\cite{2007-Gesbert-MIMO}. At the same time, different users can have different delay-aware services. For example, for users having voice services, the packet should be received with a strict delay. While for users having layered video services, the data from base layer can be transmitted and decoded earlier than the data from the enhanced layers, where the packets can be received with a long delay.
%The packet size, the constellation size and the delay constraint are some key factors to influence the energy efficiency of a communication system~\cite{2006-Taesang-Yoo-Throughput,2011-Li-EE}. A large packet is susceptible to error which may cause retransmission, so as to affect the throughput~\cite{2006-Taesang-Yoo-Throughput} and the energy efficiency. On the other hand, large constellation size requires more power for transmission with particular symbol error probability, and small constellation size can reduce the transmit power. Thus, a specific modulation has impact on the energy efficiency. When the packet size, the constellation size and the symbol error probability are given, the packet delay can be derived. Then, under different delay constraints, different throughput and transmission power are required~\cite{2011-Li-EE}, which hence influences the energy efficiency. Motivated by these issues of energy efficiency, we try to find the optimal packet size and the optimal constellation size for energy efficient communications, and investigate the optimal energy efficiency for given delay constraints.

%There exists several studies on the packet size optimization for wireless networks~\cite{1998-P-Packet,1989-Goodman-Packet}, which consider the maximization of  energy efficiency subject to transmission power or data rates. In~\cite{2004-Y-Packet,2012-M-Packet}, energy efficient packet size for sensor network has been studied. However, these works did not consider the inherent relation between the packet size and the delay aware energy efficiency.
\subsection{Related Work}
At the physical~(PHY) layer, energy efficient communication techniques are mainly developed through coding, modulation, and signal processing techniques~\cite{2005-S-Cui-Modulation,2011-Costa-EE,2014-Wu-EE}. Thus, the modulation for the MIMO spatial streams has great impact on the user's energy efficiency. The authors in~\cite{2005-S-Cui-Modulation} study the optimal modulation in multi-hop time division multiple access (TDMA) networks, who use the convex-optimization method to minimize the energy consumption per bit under the delay constraint, but different modulation sizes for each user's stream is not considered. The constellation size for each stream of a MIMO systems can influence the energy efficiency~\cite{2004-S-Cui-MIMO}, and this work shows that the energy efficiency is dramatically increased with the optimal constellation size. However, the influence of the constellation size to the power consumption of the MIMO streams has not been considered there, although the optimal power allocation for MIMO spatial streams has been studied in~\cite{2013-Miao-MIMO}. In addition, MIMO systems are not always superior to the single input and multiple output~(SIMO) systems due to different circuit power consumption~\cite{2011-Li-EE}. There exists multiple circuits in multiple transmission antennas, such as mixers, synthesizers, digital-to-analog converters, filters, etc. Hence the circuit power consumption of MIMO systems is higher than that of SIMO~\cite{2004-S-Cui-MIMO}. Therefore, each user should choose a better transmission mode between MIMO and SIMO to improve the energy efficiency~\cite{2009-Kim-MIMO,2004-S-Cui-MIMO}.
For the upper-layer service, different delay demands of the service can influence the energy efficiency~\cite{2002-Gallager-Delay,2004-Rajan-Delay,2007-Cui-Cross-layer,2012-Leila-EE}. The major existing works focus on the tradeoff between the average delay and the average transmission power, not considering the MIMO/SIMO mode switching. In~\cite{2009-Kim-MIMO}, the delay aware MIMO/SIMO switching strategy is proposed, however, the strategy is based on the flow delay and doesn't include the optimization of the constellation size. MU-MIMO has not been considered in~\cite{2009-Kim-MIMO} either.
%Since multiple antennas can avoid the interference and improve QoS~\cite{2012-Boccardi-MIMO}, which can enhance the rate capacity in comparison with the single antenna systems.
As we have stated, green communication is a major theme of 5G networks~\cite{2014-Lin-5G}, it is preferable to minimize the transmit power under different delay demands. Then it motivates us to consider total power minimization under individual delay demand in MU-MIMO cellular networks.

Since the different ratios of delay sensitive users to delay tolerant users cause the different energy efficiency of MU-MIMO and MU-SIMO, it is necessary to consider the MU-MIMO/SIMO switching. In~\cite{2007-Han-delay}, different types of services such as voice, data, and multimedia, as well as different delay constraints are considered by the user's scheduling schemes, which explore the time, channel, and multiuser diversity to guarantee QoS and enhance the network performance. In~\cite{2009-Poor-EE,2009-Farhad-EE}, joint power and rate control have been studied extensively for multiple users network. The authors have studied joint power and rate control under bit error rate~(BER) and delay constraints. However, all these works are only based on the single stream.
%On the other hand, long delay may cause buffer overflow, which has not been considered in the area of energy efficiency in most of the recent research works. In fact, buffer overflow can be solved by congestion control algorithm, which is the main function of the transmission control protocol~(TCP). Note that the congestion problem has been considered in~\cite{2007-Atilla-Eryilmaz-Congestion}, and the congestion information is transmitted back to the source by putting the congestion price in the acknowledgement~(ACK) packets. Therefore one can utilize the congestion control algorithm to solve the buffer overflow in energy efficiency by taking the queue length in the buffer of the user as the congestion price.

\subsection{Main Contributions}
In a broad view, our work considers the cross-layer design framework, which aims to take upper-layer delay-aware traffic and physical layer transmission schemes into account. Our work is related to the works in~\cite{2004-S-Cui-MIMO,2013-Miao-MIMO,2009-Kim-MIMO,2009-Poor-EE}, however, our work is different in several significant ways. First, in our work, we derive the modulation of the multiple streams to maximize the energy efficiency, but \cite{2013-Miao-MIMO,2009-Poor-EE} concentrate more on the power allocation. Meanwhile, since MU-SIMO systems may be more energy efficient than the MU-MIMO systems when the total number of users with heterogeneous delay-aware services is different, we propose an adaptive MU-MIMO/MU-SIMO transmission strategy to improve the energy efficiency in MU-MIMO systems, and select the optimal antenna for MU-SIMO mode, which are not considered in~\cite{2004-S-Cui-MIMO,2009-Kim-MIMO}. To our best knowledge, the modulation size and the antenna selection for the delay-aware energy efficiency, has not been considered in MU-MIMO systems so far, and the prior works in this area did not explicitly take into account the effect of the heterogeneous packet delay constraints for different users to the MU-MIMO/MU-SIMO switching.
%we can use adaptive modulation to improve the data rate of wireless systems by adapting some design parameters to the time-varying channel.
%The goal is to maximize the data rate of transmission while meeting the needs of BER~\cite{1995-Webb-Variable-rate,1997-Goldsmith-Variable-rate,1998-Goldsmith-Adaptive-modulation,1996-Matsuoka-Adaptive-modulation}.  %The throughput here is defined as the data rate received successfully.
%In~\cite{1996-M-L-Honig-Allocation,1999-S-J-Oh-Adaptive}, joint power and rate control have been studied extensively for multiple users networks, the authors have studied joint power and rate control under BER and delay constraints.

%Due to the burst service and the long delay, there will be congestion causing buffer overflow in the user, which can influence the energy efficiency. To deal with this problem, we use the congestion control algorithm and the virtual queue algorithm to offset its influence, and then we can see the superiority and the convergence property of the algorithms.
%In all, we consider the cross-layer model in the user station, from which the optimal energy efficiency solution can be asymptotically obtained. The energy-delay tradeoffs in wireless networks have been studied in various works~\cite{2000-Berry-delay,2002-Gallager-Delay,2004-Rajan-Delay,2009-Kim-MIMO}. But the energy-delay tradeoff investigated in this paper is based on energy efficient transmission in MIMO systems considering the circuit power.

In all, our contribution can be summarized as follows:
\begin{itemize}
 \item We obtain, via the minimum mean square error~(MMSE) receiver, the closed-form expression of the transmit power for each stream, which is related to the modulation for each stream and the symbol error rate.
 \item We derive the closed-form expression of the energy efficient modulation size for each stream, under the objective to minimize the total transmission power. Correspondingly, the total average power consumption of MU-MIMO is also obtained.
 \item In order to minimize the total power consumption of MU-MIMO, we consider the antenna selection for each user, and select the antenna which has the best channel gain.
 \item Based on the power model of the MU-MIMO and MU-SIMO, we derive a energy efficient switching policy, which considers the ratio of the number of delay sensitive users to that of delay tolerant users.
\end{itemize}

%For a
%MIMO relaying BC, there are two independent channel links
%between source and receivers; i.e., \emph{source-relay-receivers } links, and
%\emph{source-receivers} direct links~(DLs).

%
%
\subsection{Paper Organization}
The rest of the paper is organized as follows. Section \uppercase\expandafter{\romannumeral2} describes the system model, including the transmission model and the queuing model. Section \uppercase\expandafter{\romannumeral3} describes the energy efficiency and the problem statement. In section \uppercase\expandafter{\romannumeral4}, we analyze the energy efficient modulation size for each stream of the MU-MIMO systems based on the MMSE receiver. Section \uppercase\expandafter{\romannumeral5} analyzes the delay performance and the mode switching between the MU-MIMO and the MU-SIMO based on the heterogeneous delay-aware services. In section \uppercase\expandafter{\romannumeral6}, we show the simulation results, and the conclusions are made in section~\uppercase\expandafter{\romannumeral7}.

\subsection{Notation}
\emph{Notations:} $\textsf{E}(\cdot)$, $\left \|\mathbf{\cdot}\right \|_F$, $\left \|\mathbf{\cdot}\right \|_2$ and $(\cdot)^{H}$, denote the expectation, the Frobenius norm, the Euclidean norm and the conjugate transpose, respectively. $i.i.d.$  stands for independent and identically distributed. $\mathbf{I}$ is the identity matrix with appropriate dimensions.
$\mathrm{diag}(\cdot)$ is a diagonal matrix. $\lfloor x\rfloor$ is the largest integer number that is not larger than $x$, and  $\lceil x\rceil$ is the smallest integer number that is not smaller than $x$. $\bar a$ is the conjugate of $a$. For the matrices $\mathbf{A}$ and $\mathbf{B}$, $\mathbf{A}\bigoplus\mathbf{B}$ stands for the diagonal block matrix with $\mathbf{A}$ and $\mathbf{B}$ as the diagonal entries.

\begin{figure}[!t]
\begin{center}
\includegraphics [width=3.5in]{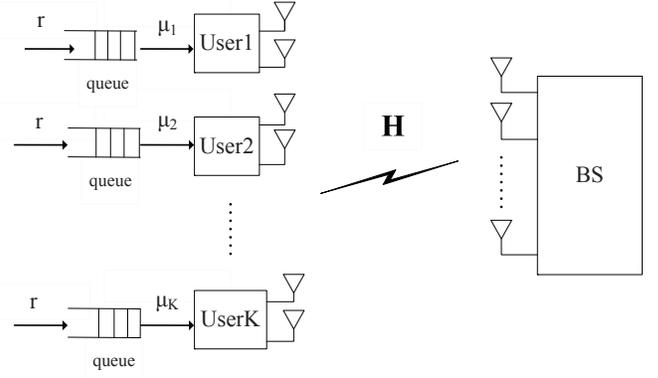}
\caption{System model.} \label{Fig-system_model}
\end{center}
\end{figure}
%
%$\log$ is of base $2$. $\mathcal{C}^{M\times N}$ represents the set of $M\times N$ matrices over complex field, and $\sim\mathcal{CN}(x,y)$ means satisfying  a circularly symmetric complex Gaussian distribution with  mean $x$ and covariance $y$. $\mathcal{U}=\{1,2,\cdots,K\}$.
%

%
\section{System Model}
\subsection{Physical Layer Channel Model}
Consider the uplink multi-users MIMO~(MU-MIMO) systems, as illustrated in Fig.~\ref{Fig-system_model}, where one base station~(BS) is serving $K$ users, we assume the channel is an independent and identically distributed~(i.i.d.) Rayleigh fading channel. The BS has $N_r$ antennas, and each user has $N_t$ antennas. Among the $K$ users, the heavy users occupy the ratio of $\rho$, and $\rho \leq 1$, where the heavy users are the users having the delay sensitive sessions, and the remaining users have the delay tolerant sessions. Denote $\mathbf{H}_i$ and $\mathbf{P}_i=\mbox{diag}\left\{\sqrt{p_{i1}},\sqrt{p_{i2}}, \cdots, \sqrt{p_{iN_t}}\right\}$ as the channel matrix and the power allocation matrix of user $i$ respectively. The total transmit power of the MU-MIMO is $P$. In a flat-fading propagation environment, the received signal at the BS is denoted as
\begin{equation}\label{receiveds}
\mathbf{y}=\mathbf{H}\mathbf{P}\mathbf{x}+\mathbf{n}=\sum_{i=1}^{K}\mathbf{H}_i\mathbf{P}_i\mathbf{x}_i+\mathbf{n},
\end{equation}
where $\mathbf{P}=\mathrm{diag}\left\{\mathbf{P}_1,\cdots,\mathbf{P}_K\right\}$, $\mathbf{y}=[y_1, y_2, \cdots, y_{N_r}]^T$ and $\mathbf{x}_i=[x_{i1}, x_{i2}, \cdots, x_{iN_t}]^T$ are the received and transmitted symbols of user $i$ respectively, and each element $x_{ij}$ can come from a $2^{b_{ij}}$-QAM modulation and is subject to a unit power constraint $E\left [ \left | x_{ij} \right |^2 \right ]=1$. $\mathbf{n}$ is the length-$N_r$ noise vector, which is Gaussian distributed with zero mean and the covariance matrix $\sigma^2\mathbf{I}$. We assume that Nyquist pulses are used and hence the M-ary quadrature amplitude modulation~(MQAM) symbol rate is approximately equal to the transmission bandwidth $B$. The total number $b$ of the information bits that can be transmitted at each time for each user $i$ is assumed the same and given by $$b=\sum_{j=1}^{N_t}b_{ij},$$ where $b_{ij}$ is the modulation size of data stream $j$ of user $i$.

The channel state information~(CSI) is supposed to be imperfectly known to the receiver. Denote $\mathbf{\hat{H}}$ as the estimated CSI at the receiver~(CSIR). Thus the channel can be modeled as~\cite{2011-Andrews-MIMO}
\begin{equation}
\mathbf{H}=\sqrt{1-\tau^2}\mathbf{\hat{H}}+\tau\mathbf{\Omega},
\end{equation}
where $\mathbf{\Omega}$ has $i.i.d$ entries of zero mean and unit variance independent of $\mathbf{\hat{H}}$ and $\mathbf{n}$, and the parameter $\tau\in [0,1]$ reflects the estimation accuracy or quality of the channel $\mathbf{H}$. The case $\tau=0$ corresponds to perfect CSIR, whereas for $\tau=1$, the CSIR is completely unknown to the receiver.

\subsection{Link Layer Queuing Model}
The transmitted bits at the physical layer come from the link layer in a packet basis. Each packet has a size of $L$ bits, among which $L_h$ bits are the header, and then
$$L_p=L-L_h$$
are the payload bits. When the receiver correctly receives a packet, it will feedback an ACK packet to the transmitter, which is the signal passed between transmitter and receiver to signify acknowledgement or receipt of response. If the receiver can not correctly receive the packet, the transmitter will repeat transmitting the packet until it is received correctly.

For each user, the link layer packets arrive at the transmitter into a first-in-first-out~(FIFO) queue, and the buffer is finite with size $Q_0$. Consider that each user's link layer constructs packet streams with the packet size of $L$ bits. With regard to the delay performance of the packet, assume that each user's queuing model is a single server M/G/1 queue~\cite{2007-Atilla-Eryilmaz-Congestion}, as shown in Fig.~\ref{Fig-system_model}. The mean packet generation rate from the data link layer is $r$ for each user, and the mean service rate at the physical layer is $\mu$. Clearly, the service rate $\mu$ depends on $b$, and the total number of the transmitted bits through the channel at each time is determined by the channel model in~\eqref{receiveds}.
 %From the congestion control, we can get the stable queue-length to avoid congestion in the use's queue. %And we have introduced a virtual queue model in congestion control algorithms, we will see this can improve the energy-efficient with delay constraint.

%\begin{IEEEeqnarray} {rCl}
%\mathbf{y}_{1k}&=&\mathbf{H}_{kb}\mathbf{P}_{k}\mathbf{s}_{k}+\sum^{K}_{j=1,j\neq k}\mathbf{H}_{kb}\mathbf{P}_{j}\mathbf{s}_{j}+  \mathbf{n}_{1k},\\
%\mathbf{y}_{r}&=&\mathbf{H}_{rb}\mathbf{P}\mathbf{s}+ \mathbf{n}_{r},\\
%\mathbf{y}_{2k}&=&\mathbf{H}_{kr}\mathbf{F}\mathbf{H}_{rb}\mathbf{P}_{k}\mathbf{s}_{k}+
%\sum^{K}_{j=1,j\neq k} \mathbf{H}_{kr}\mathbf{F}\mathbf{H}_{rb}\mathbf{P}_{j}  \mathbf{s}_{j}+\mathbf{H}_{kr}\mathbf{F}\mathbf{n}_{r}+\mathbf{n}_{2k}
%\end{IEEEeqnarray}
%
%
%\begin{figure*}[ht]

%In the following section, we first set up the energy-efficient expression and then show non-cooperative games in which actions open to each user are the choice of power consumption as well as the constellation size, at the end, we propose a congestion control algorithm to improve the energy-efficiency with delay constraint.
%

%
%\begin{figure}[!t]
%\begin{center}
%\includegraphics [width=3.3in]{queue.eps}
%\caption{User queue model based on M/G/1 queue.} \label{Fig-queue}
%\end{center}
%\end{figure}
%

%
\section{Energy Efficiency and Problem Statement}

\subsection{Throughput Analysis}
Assume that each user's packet contains the same $L$ bits, which is transmitted with $N_t$ streams. Define the transmission time per packet as
\begin{IEEEeqnarray}{rCl}
%\begin{aligned}
t_{L}=\frac{L}{R_{s}b}.
%\end{aligned}
\end{IEEEeqnarray}
Since not all the transmitted data in the packet are information bits, we define the effective throughput $T_i$ for user $i$ as the payload information that can be correctly received per second as~\cite{2000-Pursley-PSK,2001-Lavery-thesis,2002-Catreux-Throughputs}:
\begin{IEEEeqnarray}{rCl}
T_i=\frac{L_{p}p_{s}}{t_{L}}=\frac{L-L_{h}}{L}\sum_{j=1}^{N_{t}}b_{ij} R_{s}p_{s}=\frac{L-L_{h}}{L}b R_{s}p_{s}.\label{throughput}
\end{IEEEeqnarray}
where $p_{s}$ is the probability of successful packet transmission for user $i$ at the link layer. To facilitate the analysis of packet throughput, $p_{s}$ needs to be derived. Based on the relationship of the packet and the symbol, $p_{s}$ can be expressed as a function of the symbol error rate~(SER) $p_{e}$ for each data stream of user $i$. Since delay is related to $p_{s}$, in order to derive the closed-form expression of delay, we assume that the SER $p_e$ is given for each stream. Then the throughput is given based on the given SER.

\subsection{Power Consumption and Energy Efficiency}
A power consumption model is required to evaluate the energy efficiency for any communication system. In this paper, we only consider the power consumption of the transmit side for simplicity. To realize the system throughput $T=\sum_{i=1}^K T_i$, the total power consumption of the MU-MIMO systems consists of the total transmit power $P$ and the total circuit power $P_c$. The circuit power consumption is modeled as a linear function of the number of the transmit antennas, and the circuit power for each antenna is $P_0$. This overly simplified model has been widely adopted in the analysis of energy efficiency~\cite{2015-wkl-MIMO,2007-Cui-Cross-layer}. For the MU-SIMO systems, each user chooses the antenna with the best channel gain, and the other antennas are not used. The circuit operation can be turned off for unused antennas. Then each user has circuit power consumption of $P_0$.

We define the energy efficiency as the number of transmitted bits per unit energy consumption, which is equivalent to the throughput per total power consumption. The energy efficiency of the MU-MIMO systems can be shown as
\begin{equation}
f_{ee}\triangleq T/(P+P_c).\label{eedefinition}
\end{equation}

\subsection{Problem Statement}
In this paper, we are interested in determining the energy efficient modulation $b_{ij}$ for each stream of the MU-MIMO systems. In addition, since total power consumption of MU-MIMO amd MU-SIMO are different for different ratios $\rho$, we should switch between MU-MIMO and MU-SIMO. In particular, while the heavy users occupy the ratio $\rho=\rho_0$, we want to get the energy efficient transmission mode from MU-MIMO and MU-SIMO, which can also be called the switching strategy $t$. Mathematically, the problem is given by
\begin{equation}\label{eeproblem}
\max_{\left\{b_{ij}\right\},t}\left\{f_{ee}| \rho=\rho_0, t\in \left\{m,s\right\} \right\},
\end{equation}
where $m$ and $s$ represent the transmission modes of MU-MIMO and MU-SIMO respectively.

\section{Energy Efficiency for Physical Layer Transmission}
Since the total circuit power of the MU-MIMO is a constant and the SER for each stream is given, the system throughput $T$ is derived based on~\eqref{throughput}. Then the problem of getting the energy efficient modulation of MU-MIMO is equivalent to minimizing the total transmit power at the same throughput.

Under the assumption of the imperfect CSI at the base station, this section introduces the deterministic approximation of the signal to interference and noise ratio~(SINR) in MU-MIMO system, which is based on the minimum mean square error~(MMSE) receiver. To get the energy efficient modulation for each stream of each user, we should solve the optimization problem of power minimization. The derived results will be used in the next section.

\subsection{The Derivation of the Receiving SINR}
Consider the MMSE detecting matrix
\begin{equation}\nonumber
\mathbf{\hat{W}}_{MMSE}=(\mathbf{P}^H\mathbf{\hat{H}}^H\mathbf{\hat{H}}\mathbf{P}+\alpha\mathbf{I}_{N_tK})^{-1}\mathbf{P}^H\mathbf{\hat{H}}^H,\label{detector}
\end{equation}
where $\alpha=\frac{\sigma^2}{P}$ is the regularizing factor in the MMSE receiver. Apparently, substituting $\mathbf{\hat{W}}_{MMSE}$ into~\eqref{receiveds} to derive the received SINR is complex. To simplify the calculation, we extract the power allocation matrix $\mathbf{P}$ to get
\begin{equation}
\mathbf{\hat{W}}=(\mathbf{\hat{H}}^H\mathbf{\hat{H}}+\alpha\mathbf{I}_{N_tK})^{-1}\mathbf{\hat{H}}^H.
\end{equation}

As can be seen, the optimal receiving filter matrix $\mathbf{\hat{W}}_{MMSE}$ depends on the power allocation matrix $\mathbf{P}$, but the simplified $\mathbf{\hat{W}}$ no longer depends on $\mathbf{P}$, which can reduce the computational complexity. However, this influence to the modulation allocation for each stream is small, which is validated by Fig.~\ref{Fig-influence}.
%Therefore, directly solving the power allocation matrix is complex. In order to solve this problem, we propose an algorithm which iteratively computes $\mathbf{\hat{W}}$ and $\mathbf{P}$.
Therefore, the receiving signal vector processed by an MMSE detector can be denoted as
\begin{IEEEeqnarray}{rCl}
\begin{aligned}
\mathbf{z}=\mathbf{\hat{W}}\mathbf{y}&=\mathbf{\hat{W}}\mathbf{H}\mathbf{P}\mathbf{x}+\mathbf{\hat{W}}\mathbf{n}\\
&=(\mathbf{\hat{H}}^H\mathbf{\hat{H}}+\alpha\mathbf{I}_{N_tK})^{-1}\mathbf{\hat{H}}^H\mathbf{H}\mathbf{P}\mathbf{x}+\mathbf{\hat{n}},\label{MMSE_p}
\end{aligned}
\end{IEEEeqnarray}
where $\mathbf{\hat{n}}=(\mathbf{\hat{H}}^H\mathbf{\hat{H}}+\alpha\mathbf{I}_{N_tK})^{-1}\mathbf{\hat{H}}^H\mathbf{n}$.
%\begin{equation}
%\begin{aligned}
%\mathbf{z}(n+1)&=\mathbf{\hat{W}}(n+1)\mathbf{y}\\
%&=(\mathbf{\hat{H}}^H\mathbf{\hat{H}}+\alpha(n)\mathbf{I}_{N_t})^{-1}\mathbf{\hat{H}}^H\mathbf{H}\mathbf{P}(n+1)\mathbf{x}+\mathbf{\hat{n}}
%\end{aligned}
%\end{equation}

After linear MMSE, the vector $\mathbf{z}$, is the linear MMSE estimate of the transmitted symbols $\mathbf{x}$. This can also be interpreted as a linear equalizer, which can reduce the inter-symbol interference~(ISI) due to the parallel transmission of independent symbols over the nonorthogonal radio channel. Therefore, MMSE receiver has been widely used in current MU-MIMO systems to improve the performance.
%rather than the ISI among symbols transmitted at different time epoch in the single channel systems.

It can be easily shown (see, e.g., ~\cite{1998-Verdu-Multiuser}) that the instantaneous received SINR for the $i$-th filter output is corresponding to the $i$-th element of $\mathbf{z}$. To evaluate the amount of the desired signal and interference on each spatial stream by MMSE filter, we use the unitary decomposition $\mathbf{\hat{H}^H}\mathbf{\hat{H}}=\mathbf{Q}\mathbf{\Lambda}\mathbf{Q}^H$ with a nonnegative diagonal eigenvalue matrix $\mathbf{\Lambda}=\mbox{diag}\left\{\lambda_{1},\cdots,\lambda_{N_tK}\right\}$ and an unitary eigenvector matrix $\mathbf{Q}$, and find
\begin{equation}
\begin{aligned}
\mathbf{\hat{W}}\mathbf{H}\mathbf{P}\mathbf{x}&=\sqrt{1-\tau^2}(\mathbf{\hat{H}^H}\mathbf{\hat{H}}+\alpha\mathbf{I}_{N_tK})^{-1}\mathbf{\hat{H}}^H\mathbf{\hat{H}}\mathbf{P}\mathbf{x}+\tau\mathbf{\hat{W}}\mathbf{\Omega}\mathbf{P}\mathbf{x}\\
&=\sqrt{1-\tau^2}\mathbf{Q}\frac{\mathbf{\Lambda}}{\mathbf{\Lambda}+\alpha \mathbf{I}_{N_tK}}\mathbf{Q}^H\mathbf{P}\mathbf{x}+\mathbf{G}_e,\label{signal}
\end{aligned}
\end{equation}
where $\mathbf{G}_e=\tau\mathbf{\hat{W}}\mathbf{\Omega}\mathbf{P}\mathbf{x}$ is the noise from the channel estimation error. Using~\eqref{signal}, we can find the entry of the $j$-th spatial stream of user $i$ as~\eqref{smatrix}.
\begin{figure*}[!t]
\normalsize
\begin{equation}
\begin{aligned}
s_{ij}=[q_{(i-1)K+j,1}\frac{\lambda_{1,1}\sqrt{1-\tau^2}}{\lambda_{1,1}+\alpha} \quad \cdots \quad q_{(i-1)K+j,N_t K}\frac{\lambda_{K,N_t}\sqrt{1-\tau^2}}{\lambda_{K,N_t}+\alpha}]\\
\times \begin{bmatrix}
\bar{q}_{1,1}\quad \cdots & \quad \bar{q}_{N_tK,1}\\
\vdots \qquad & \vdots  \\
\bar{q}_{1,N_tK} \quad \cdots & \quad \bar{q}_{N_tK,N_tK}
\end{bmatrix}
\begin{bmatrix}
\sqrt{p_{11}}x_{11}\\
\vdots\\
\sqrt{p_{KN_t}}x_{KN_t}
\end{bmatrix}.\label{smatrix}
\end{aligned}
\end{equation}
\hrulefill
\vspace*{10pt}
\end{figure*}
To find the expected power of the desired signal, we compute the expectation over $\mathbf{Q}$. From~\eqref{smatrix} and ~\cite{2005-RZF}, we can get the desired signal power of the stream $j$ of user $i$
\begin{equation}\label{pi}
\begin{aligned}
E(P_{s_{ij}})&=p_{ij}(1-\tau^2)E\left ( \sum_{l=1}^{N_tK}\frac{\lambda_l}{\lambda_l+\alpha}|q_{(i-1)K+j,l}|^2 \right )^2\\
&=\frac{p_{ij}(1-\tau^2)}{N_tK(N_tK+1)}\left [ \left ( \sum_{l=1}^{N_tK}\frac{\lambda_l}{\lambda_l+\alpha} \right )^2+\sum_{l=1}^{N_tK}\left ( \frac{\lambda_l}{\lambda_l+\alpha} \right )^2 \right ],
\end{aligned}
\end{equation}
where the expectation is taken with respect to distribution $\mathbf{Q}$ conditioned on $\mathbf{\Lambda}$. Note that conditional expectation taken with respect to $\mathbf{Q}$ that is conditioned on $\mathbf{\Lambda}$ is valid because $\mathbf{Q}$ and $\mathbf{\Lambda}$ are independent~\cite{Edelman-1989-matrices}. All of the remaining terms in~\eqref{smatrix} are the interference for stream $j$.

To find the expected power of the interference, we introduce the following lemma:
\newtheorem{lemma}{Lemma}
\begin{lemma}\label{p1}
If $1\leq i,j,i^{'},j^{'}\leq N_{t}K,i\neq i^{'},j\neq j^{'}$, and $\mathbf{Q}=[q_{ij}]_{N_tK\times N_tK}$ is a standard unitary matrix, then
$$E\left (q_{ij}q_{ij^{'}}\bar{q}_{i^{'}j}\bar{q}_{i^{'}j^{'}} \right )=\frac{-1}{N_tK(N^2_tK^2-1)}.$$
\end{lemma}
\begin{proof}[Proof]
The proof of Lemma~\ref{p1} is given in Appendix~\ref{a1}.
\end{proof}

From~\eqref{smatrix}, we can get the interference for the $i$-th filter output as
\begin{equation}
I_i=\sum_{n=1,n\neq i}^{N_tK}\sum_{m=1}^{N_tK}q_{i,m}\bar{q}_{n,m}\frac{\lambda_m\sqrt{1-\tau^2}}{\lambda_m+\alpha}\sqrt{p_n}x_n.
\end{equation}
Therefore, by Lemma~\ref{p1}, we can get the expected power of the interference as
\begin{equation}\label{I}
\begin{aligned}
E(P_{I_{i}})=\sum_{n=1,n\neq i}^{N_tK}\sum_{m=1}^{N_tK}p_nE\left (\left | q_{i,m} \right |^2\left | \bar{q}_{n,m} \right |^2 \right )\frac{\lambda^2_{m}(1-\tau^2)}{(\lambda_m+\alpha)^2}\\
=-\frac{(P-p_{ij})(1-\tau^2)}{N_tK(N^2_tK^2-1)}\left [ \left ( \sum_{l=1}^{N_tK}\frac{\lambda_l}{\lambda_l+\alpha} \right)^2 \right.\\
\left.-N_tK\sum_{l=1}^{N_tK}\left (\frac{\lambda_l}{\lambda_l+\alpha} \right )^2\right ].
\end{aligned}
\end{equation}

From~\eqref{MMSE_p} and~\eqref{signal}, the covariance of the noise is shown in~\eqref{noise},
\begin{figure*}[!t]
\normalsize
\begin{equation}\label{noise}
\begin{aligned}
E\left \{ \left ( \tau\mathbf{\hat{W}}\mathbf{\Omega}\mathbf{P}\mathbf{x}+\mathbf{\hat{W}}\mathbf{n} \right )\left ( \tau\mathbf{\hat{W}}\mathbf{\Omega}\mathbf{P}\mathbf{x}+\mathbf{\hat{W}}\mathbf{n} \right )^H \right \}\\
=\tau^2E\left \{\mathbf{\hat{W}}\mathbf{\Omega}\mathbf{P}\mathbf{x}\mathbf{x}^H\mathbf{P}^H\mathbf{\Omega}^H\mathbf{\hat{W}}^H  \right \}+E\left \{ \mathbf{\hat{W}}\mathbf{n}\mathbf{n}^H\mathbf{\hat{W}}^H \right \}\\
=\tau^2\sum_{l=1}^{N_tK}\frac{\lambda_l}{(\lambda_l+\alpha)^2}\mbox{diag}\left \{p_{11},p_{12},\cdots,p_{KN_{t}} \right \}+\mbox{tr}\left(\frac{\sigma^2\mathbf{\Lambda}}{(\mathbf{\Lambda}+\alpha\mathbf{I}_{N_tK})^2}\right)\mathbf{I}_{N_tK},
\end{aligned}
\end{equation}
\hrulefill
\vspace*{10pt}
\end{figure*}
where we used $E\left \{ \mathbf{\Omega}\mathbf{A}\mathbf{\Omega}^H \right \}=\mbox{tr}(\mathbf{A})\mathbf{I}_N$ for any $N\times N$ matrix $\mathbf{A}$ in~\cite{2011-random-matrix} and $E\left \{ \mathbf{\hat{W}}\mathbf{\hat{W}}^H \right \}=\mbox{tr}\left(\frac{\mathbf{\Lambda}}{(\mathbf{\Lambda}+\alpha\mathbf{I}_{N_tK})^2}\right)\mathbf{I}_{N_tK}$ in ~\cite{2005-RZF}.

Assuming that each filter output is decoded independently, we set $f(\lambda_i)\triangleq \frac{\lambda_i}{(\lambda_i+\alpha)^2}$. From~\eqref{pi}, \eqref{I} and~\eqref{noise}, the SINR of a linear MMSE detector on the $j$-th spatial stream of user $i$ can be computed as~\eqref{SINR}, %\newcounter{mytempeqncnt}
\begin{figure*}[!t]
\normalsize
\begin{equation}
\label{SINR}
SINR_{ij}=\frac{p_{ij}(1-\tau^2)\left(f_1(\mathbf{\lambda})+f_2(\mathbf{\lambda})\right)}{(P-p_{ij})(1-\tau^2)(f_2(\mathbf{\lambda})-N_tKf_1(\mathbf{\lambda}))/(1-N_tK)+N_tK(N_tK+1)p_{ij}\tau^2f_3(\lambda)+N_tK(N_tK+1)\sigma^2f_3(\mathbf{\lambda})}.
\end{equation}
\begin{equation}
\label{pib}
p_{ij}=\frac{\eta(b_{ij})\left [ P(1-\tau^2)(f_2(\mathbf{\lambda})-N_tKf_1(\mathbf{\lambda}))/(1-N_tK)+f_3(\mathbf{\lambda})N_tK(N_tK+1)\sigma^2 \right ]}{\left ( f_1(\lambda)+f_2(\lambda) \right )(1-\tau^2)+\eta(b_{ij})\left [(1-\tau^2)(f_2(\mathbf{\lambda})-N_tKf_1(\mathbf{\lambda}))/(1-N_tK)-f_3(\lambda)\tau^2N_tK(N_tK+1) \right ]}.
\end{equation}
\hrulefill
\vspace*{10pt}
\end{figure*}
where
\begin{IEEEeqnarray}{rcl}
\begin{aligned}
f_1(\mathbf{\lambda})&\triangleq& \sum_{l=1}^{N_tK}\left ( \frac{\lambda_l}{\lambda_l+\alpha} \right )^2,\\
f_2(\mathbf{\lambda})&\triangleq& \left ( \sum_{l=1}^{N_tK}\frac{\lambda_l}{\lambda_l+\alpha} \right)^2,\\
f_3(\mathbf{\lambda})&\triangleq& \sum_{l=1}^{N_tK}\frac{\lambda_l}{(\lambda_l+\alpha)^2}.
\end{aligned}
\end{IEEEeqnarray}

\subsection{Energy Efficient Constellation Size}
We know that the SINR per symbol is
\begin{IEEEeqnarray}{rCl}
\begin{aligned}
\gamma_{x_{ij}}\triangleq SINR_{ij}\frac{B}{R_{s}},\label{gammas}
\end{aligned}
\end{IEEEeqnarray}
for $i=1,\cdots, K$ and $j=1,\cdots, N_t$, where $R_s$ is the symbol rate.

It is well known (see~\cite{2005-Goldsmith-Wireless}) that, the SER of MQAM modulation with size $2^{b_{ij}}$, is given by
\begin{IEEEeqnarray}{rCl}
p_e(b_{ij},\gamma_{x_{ij}})=2(1-2^{-b_{ij}/2}\emph{})Q \left ( \sqrt{\frac{3}{2^{b_{ij}}-1}\gamma _{x_{ij}}} \right ),\label{symbolerror}
\end{IEEEeqnarray}
where $Q(\cdot)$ is the complementary cumulative distribution function of the standard Gaussian random variable. Using the Chernoff upper bound, we can get
\begin{IEEEeqnarray}{rCl}
p_{e}(b_{ij},\gamma_{x_{ij}}) \leq \ 2(1-2^{-b_{ij}/2})e^{-\frac{3}{2^{b_{ij}}-1}\frac{\gamma _{x_{ij}}}{2}}.\emph{}\label{pe}
\end{IEEEeqnarray}
Substituting~\eqref{gammas} into~\eqref{pe}, we can derive
\begin{equation}\label{SINRb}
SINR_{ij}=\frac{2R_s(2^{b_{ij}}-1)}{3B}\mbox{ln}\frac{2(1-2^{-b_{ij}/2})}{p^{ij}_e}\triangleq \eta(b_{ij}). \end{equation}
Using~\eqref{SINR} and~\eqref{SINRb}, we can compute the closed-form expression of transmission power for the stream $j$ of user $i$ shown in~\eqref{pib}.
%On the other hand, we assume that the circuit power consumption at each transmit antenna is a constant power $P_0$.
Specifically, when all the users have delay sensitive services, i.e., $\rho=1$, we can get the expression of the total power consumption.
%\begin{equation}\label{pib}
%p_i=\frac{f(b_i)\left [ Pf_1(\lambda)(1-\tau^2)+f(\lambda_i)N_t(N_t+1)\sigma^2 \right ]}{\left ( f_1(\lambda)+f_2(\lambda) \right )(1-\tau^2)+f(b_i)\left [ f_1(\lambda)(1-\tau^2)-f_3(\lambda)\tau^2N_t(N_t+1) \right ]}
%\end{equation}
%\begin{figure*}[!t]
%\normalsize
%\begin{equation}
%\label{pib}
%\frac{f(b_i)\left [ Pf_1(\lambda)(1-\tau^2)+f(\lambda_i)N_t(N_t+1)\sigma^2 \right ]}{\left ( f_1(\lambda)+f_2(\lambda) \right )(1-\tau^2)+f(b_i)\left [ f_1(\lambda)(1-\tau^2)-f_3(\lambda)\tau^2N_t(N_t+1) \right ]}
%\end{equation}
%\hrulefill
%\vspace*{10pt}
%\end{figure*}

\newtheorem{proposition}{Proposition}
\begin{proposition}\label{p2}
Denote $c_1=\left ( f_1(\lambda)+f_2(\lambda) \right )(1-\tau^2)$, $c_2=(f_2(\lambda)-N_tKf_1(\lambda))(1-\tau^2)/(1-N_tK)-f_3(\lambda)\tau^2N_tK(N_tK+1)$ and $c_3=(f_2(\lambda)-N_tKf_1(\lambda))(1-\tau^2)/(1-N_tK)$. For $\rho=1$, the total transmission power for the MU-MIMO mode is
\begin{equation}\nonumber
P=\frac{\sum_{i=1}^{K}\sum_{j=1}^{N_t}\frac{\hat{\sigma}^2f_3(\mathbf{\lambda})}{c_2+\frac{c_1}{\eta(b_{ij})}}}{1-c_3\sum_{i=1}^{K}\sum_{j=1}^{N_t}\frac{1}{c_2+\frac{c_1}{\eta(b_{ij})}}}, \end{equation}
where $\hat{\sigma}^2=N_tK(N_tK+1)\sigma^2$. %For $\rho=0$, the total transmission power for the SIMO mode is $$P_s=\frac{\sum_{i=1}^{K}\frac{\hat{\sigma}^2f(\lambda_i)}{c_2+\frac{c_1}{f(b)}}}{1-c_3\sum_{i=1}^{i=K}\frac{1}{c_2+\frac{c_1}{f(b)}}}$$
\end{proposition}
\begin{proof}[Proof]
The proof of Proposition~\ref{p2} is given in Appendix~\ref{a2}.
\end{proof}

In all, we can get the total power consumption of the MU-MIMO systems as
\begin{IEEEeqnarray}{rCl}
\hat{P}_m=P+\sum_{i=1}^{K}\sum_{j=1}^{N_t}P_0.\label{totalp2}
\end{IEEEeqnarray}
Let $b_{min}$ be the minimum modulation size, therefore, our problem is the general fractional programming, which can be formulated as the following problem
\begin{IEEEeqnarray}{rCl}
\begin{array}{ll}\label{minimump1}
\min \quad  P=\frac{\Phi(\mathbf{b})}{\Psi(\mathbf{b})},\\
s.t. \begin{array}[t]{rcl}
    \quad b_{ij} & \geqslant & b_{min},\\
     \sum_{j=1}^{N_{t}}b_{ij} & = & b.
       \end{array}
\end{array}
\end{IEEEeqnarray}
where $\Phi(\mathbf{b})=\sum_{i=1}^{K}\sum_{j=1}^{N_t}\frac{\hat{\sigma}^2f_3(\mathbf{\lambda})}{c_2+\frac{c_1}{\eta(b_{ij})}}$ and
$\Psi(\mathbf{b})=1-c_3\sum_{i=1}^{i=K}\sum_{j=1}^{N_t}\frac{1}{c_2+\frac{c_1}{\eta(b_{ij})}}$ are continuous real-valued functions. To solve problem~\eqref{minimump1}, we introduce the following lemma.
\begin{lemma}(Jagannathan's theorem~\cite{2000-Roubi}).  \label{th1}
$\mathbf{b}^{\ast}$ is an optimal solution for~\eqref{minimump1} if and only if $\mathbf{b}^{\ast}$ is an optimal solution for
\begin{IEEEeqnarray}{rCl}\label{minimump2}
\begin{array}{ll}
\min \quad  \Phi(\mathbf{b})-P(\mathbf{b}^{\ast})\Psi(\mathbf{b}),\\
s.t. \begin{array}[t]{rcl}
    \quad b_{ij} & \geqslant & b_{min},\\
     \sum_{j=1}^{N_{t}}b_{ij} & = & b.
       \end{array}
\end{array}
\end{IEEEeqnarray}                                        \end{lemma}

It has been shown that~(see~\cite{2000-Roubi}) problem~\eqref{minimump2} exists a solution for any $\delta\in \mathcal{R}$, where $\delta=P(\mathbf{b}^{\ast})$. To find the solution of the problem \eqref{minimump2}, we can define:
\begin{equation}\label{equalp}
f(\delta)=\mbox{min}\left\{ \Phi(\mathbf{b})-\delta\Psi(\mathbf{b}): b_{ij} \geqslant b_{min}, \sum_{j=1}^{N_{t}}b_{ij}=b\right\}.
\end{equation}
Dinkelbach developed a method based on Lemma~\ref{th1} for solving non-linear fractional problems where the function $\Psi$ is concave and $\Phi$ is convex~\cite{2000-Roubi}.

We can easily prove that $\Phi(\mathbf{b})$ and $\Psi(\mathbf{b})$ are convex and concave respectively. Thus
at the $k$-th step, we can write~\eqref{equalp} as
\begin{equation}\label{minimum_k}
f(\delta_k)=\mbox{min}\left\{ \Phi(\mathbf{b})-\delta_k\Psi(\mathbf{b}): b_{ij} \geqslant b_{min},
\sum_{j=1}^{N_{t}}b_{ij}=b\right\},
\end{equation}
which is a convex optimization. Let $\kappa$ and $\nu$ denote the Lagrange multipliers associated with the constraints in the optimization
problem~\eqref{minimum_k}.
The Lagrangian function is then given by
\begin{multline}
\mathcal{L}(\mathbf{b},\mathbf{\kappa},\nu)=\Phi(\mathbf{b})-\delta_k\Psi(\mathbf{b})  \\
-\sum_{j=1}^{N_{t}}\kappa_{j}(b_{ij}-b_{min})+\nu\left(\sum_{j=1}^{N_{t}}b_{ij}-b\right).
\end{multline}
The necessary and sufficient conditions for optimality are given by Karush-Kuhn-Tucker~(KKT) conditions~\cite{2004-Boyd-CO},
\begin{IEEEeqnarray}{rCl}
\left\{\begin{matrix}b_{ij}^{\ast}\geq b_{min},\\
\sum_{j=1}^{N_{t}}b_{ij}^{\ast}=b,\\
\kappa_{j}^{\ast}\geq 0,\\
\kappa_{j}^{\ast}(b_{min}-b_{ij}^{\ast})=0,\quad j=1,\cdots ,N_{t},\\
\frac{\partial \Phi(\mathbf{b})}{\partial b^{\ast}_{ij}}-\delta_k\frac{\partial \Psi(\mathbf{b})}{\partial b^{\ast}_{ij}}-\kappa_{j}^{\ast}+\nu^{\ast}=0,\quad j=1,\cdots,N_{t},
\end{matrix}\right.\label{kkt}
\end{IEEEeqnarray}
where $\beta_{j}\triangleq\frac{2\mbox{ln}2}{3}\mbox{ln}\frac{2}{p^{ij}_{e}}\frac{N_{0}R_{s}}{\lambda_{(i-1)N_t+j}^{2}}$, $\kappa_{j}^{\ast}$ and $\nu^{\ast}$ denote the optimal multipliers, $b_{ij}^{\ast}$ is the optimal $b_{ij}$. We can directly solve the equations in~\eqref{kkt} to find $b_{ij}^{\ast}$, $\kappa_{j}^{\ast}$ and $\nu^{\ast}$. Thus we have
\begin{IEEEeqnarray}{rCl}
b_{ij}^{\ast}=\left\{\begin{matrix}\mbox{log}_{2}\left ( \frac{-\alpha-\nu^{\ast}}{\beta_{j}} \right ),
 &\nu ^{\ast}\leq -\beta_{j}2^{b_{min}},   \\
 b_{min},& \nu ^{\ast}>  -\beta_{j}2^{b_{min}}.
\end{matrix}\right.\label{bform}
\end{IEEEeqnarray}

Now we come to the stage to propose an iterative algorithm to efficiently solve the problem~\eqref{minimump1} in Algorithm~1.
\begin{algorithm}[htb]
\caption{The framework of the iterative algorithm for problem~\eqref{minimump1}}
\begin{algorithmic}[1]\label{alg1}
 \STATE{Step 1}: Let $\mathbf{b}_1$ be a feasible point of~\eqref{minimump1} and $\delta_1=P(\mathbf{b}_1)=\frac{\Phi(\mathbf{b}_1)}{\Psi(\mathbf{b}_1)}$. Let $k=1$.
 \STATE{Step 2}: By means of convex programming to solve the following problem:
 \begin{IEEEeqnarray}{rCl}\nonumber
\begin{array}{ll}
\min \quad  \Phi(\mathbf{b})-\delta_k\Psi(\mathbf{b}),\\
s.t. \begin{array}[t]{rcl}
    \quad b_{ij} & \geqslant & b_{min},\\
     \sum_{j=1}^{N_{t}}b_{ij} & = & b.
       \end{array}
\end{array}
\end{IEEEeqnarray}
 With the KKT conditions from~\eqref{kkt}, we denote any solution point by $\mathbf{b}_{k+1}$.
 \STATE{Step 3}: If the solution $f(\delta_k)=0$, stop and $\mathbf{b}_k$ is optimal. Otherwise, set $\delta_{k+1}=P(\mathbf{b}_{k+1})=\frac{\Phi(\mathbf{b}_{k+1})}{\Psi(\mathbf{b}_{k+1})}$, and $k=k+1$, and go to step 2.
\end{algorithmic}
\end{algorithm}
Algorithm~\ref{alg1} either terminates in a finite number of iterations or it generates an infinite sequence $\left\{ \mathbf{b}_{k}\right\}$ such that any accumulation point solves~\eqref{minimump1}~\cite{2000-Roubi}.

%We can get the convergence property of algorithm~\ref{alg1}.
Thus, we can get the optimal solution $\mathbf{b}^{\ast}$.
Since $\mathbf{b}$ is an integer number, we choose the energy efficient constellation size $\mathbf{b}^{opt}$ as
\begin{equation}\nonumber
\mathbf{b}^{opt}\triangleq\arg\min_{\mathbf{b}\in\{\lfloor{\mathbf{b}^\ast}\rfloor,\lceil{\mathbf{b}^\ast}\rceil\}}{|\mathbf{b}-\mathbf{b}^\ast|}.
\end{equation}
Substitute $\mathbf{b}^{opt}$ into \eqref{totalp2}, we can get the energy efficient total power consumption of the MU-MIMO systems as $\hat{P}^{opt}_m$.

%\subsection{Energy-Efficiency with Zero-Forcing Receiver}
%For $\alpha=0$, the MMSE receiver matrix in~\eqref{detector} reduces to the ZF detecting matrix $\mathbf{\hat{W}}_{\mbox{zf}}$ which reads
%\begin{equation}\label{ZF}
%\mathbf{\hat{W}}_{\mbox{zf}}=(\mathbf{\hat{H}^H}\mathbf{\hat{H}})^{-1}\mathbf{\hat{H}}^H,
%\end{equation}
%therefore, the receiving signal vector takes the form
%\begin{multline}\label{zf-r}
%\mathbf{z}=\mathbf{\hat{W}}_{\mbox{zf}}\mathbf{y}=(\mathbf{\hat{H}^H}\mathbf{\hat{H}})^{-1}\mathbf{\hat{H}}^H\mathbf{H}\mathbf{P}\mathbf{x}+
%(\mathbf{\hat{H}^H}\mathbf{\hat{H}})^{-1}\mathbf{\hat{H}}^H\mathbf{n}\\
%=\sqrt{1-\tau^2}\mathbf{P}\mathbf{x}+\underbrace{\tau\mathbf{\hat{W}}_{\mbox{zf}}\mathbf{\Omega}\mathbf{P}\mathbf{x}+\mathbf{\hat{W}}_{\mbox{zf}}\mathbf{n}}_{\mbox{noise vector}},
%\end{multline}
%with~\eqref{noise}, \eqref{ZF} and~\eqref{zf-r}, the SNR of stream $i$ under ZF detector reads
%\begin{equation}\label{snr-zf}
%\begin{aligned}
%SNR_i=\frac{p_i(1-\tau^2)}{\tau^2\left | \mathbf{\hat{W}}_{\mbox{zf}}\mathbf{\Omega}\mathbf{P}\mathbf{x}\right|^2_i+\left|\mathbf{\hat{W}}_{\mbox{zf}}\mathbf{n} \right |^2_i}\\
%=\frac{p_i(1-\tau^2)}{p_i\tau^2\sum_{l=1}^{N_t}\frac{1}{\lambda_l}+\sigma^2/\lambda_i}.
%\end{aligned}
%\end{equation}
%
%Then, from~\eqref{SINRb} it follows that the equivalent total transmission power is given by
%\begin{equation}
%P=\sum_{i=1}^{N_t}p_i=\sum_{i=1}^{N_t}\frac{\sigma^2f(b_i)/\lambda_i}{(1-\tau^2)-\tau^2f(b_i)f(\lambda)},
%\end{equation}
%where $f(\lambda)=\sum_{l=1}^{N_t}\frac{1}{\lambda_l}$.

\section{Energy Efficiency for QoS of Delay}
In the last section, we have solved the problem of energy efficient constellation size allocation for MU-MIMO based on MMSE detector, based on which, this section will derive the energy efficient transmission mode switching between MU-MIMO and MU-SIMO under the condition of different ratios of delay sensitive users to delay tolerant users. At first, we will derive the closed-form expression of delay for each user. Then we will show how the heterogeneous delay influence the transmission mode.

\subsection{Delay Analysis}
%where $p_{s}$ is the probability of packet successfully transmission, which can be expressed in terms of the same symbol error probability $p_{e}$ from each stream.
For each user, the packet can be divided into $N_t$ streams to be transmitted by the physical layer. The number of information bits for each stream linearly scales with the data rate.
Then we can get the number of information bits from each stream $j$ in a packet as:
$$L_j=\frac{b_{ij}L}{\sum_{j=1}^{N_t}b_{ij}}.$$ So the number of symbols for each stream in a packet is $L_{j}/b_{ij}$.
With the assumption of the same SER for each stream $j$, we can get
\begin{IEEEeqnarray}{rCl}
%\begin{aligned}
p_{s}=\prod_{j=1}^{N_t}\left ( 1-p_{e} \right )^{\frac{L_j}{b_{ij}}}=\left ( 1-p_{e} \right )^{N_tL/\sum_{j=1}^{N_t}b_{ij}}.\label{packetprob}
%\end{aligned}
\end{IEEEeqnarray}

In M/G/1 queue model, the packet service time $S_T$ has the following probability mass function:
\begin{IEEEeqnarray}{rCl}
P\left \{S_T=nt_p \right \}=p_{s}(1-p_{s})^{n-1},\quad \mbox{for}\quad n=1,\cdots,\label{st}
\end{IEEEeqnarray}
where $t_p$ represents the packet transmission time when the queue is serving one packet in one time slot, which is given by
\begin{IEEEeqnarray}{rCl}
t_p=\frac{L}{bR_{s}}.\label{tau}
\end{IEEEeqnarray}
From~\eqref{st}, we can get the mean service time:
\begin{IEEEeqnarray}{rCl}
\begin{aligned}
\textsf{E}\left \{ S_T \right \}&=\sum_{n=1}^{\infty }nt_p p_{s}(1-p_{s})^{n-1}\\
&=\frac{t_p}{p_{s}}.\label{mst}
\end{aligned}
\end{IEEEeqnarray}
From~\eqref{tau} and~\eqref{mst}, the service rate $\mu_i$ for user $i$ is given by:
\begin{IEEEeqnarray}{rCl}
%\begin{aligned}
\mu_i=\frac{1}{\textsf{E}\left \{ S_T \right \}}&=\frac{p_{s}}{t_p}=\frac{bR_{s}p_{s}}{L}.
%\end{aligned}
\end{IEEEeqnarray}

%Note that the service rate $\mu$ is a constant regardless of the packet number $Q$ in the buffer. Let $r_Q$ and $\mu_Q$ be respectively the generation rate from source and the service rate when there are $Q$ packets in the buffer. Thus, the M/G/1 queue is a birth-death process with $r_{Q}=r$ ($Q\geq 0$) and $\mu_{Q}=\mu$ ($Q\geq 0$).
By~\cite{1985-Gross-Queueing}, using the Pollaczek-Khintchine formula, we can get the mean queue length as
\begin{IEEEeqnarray}{rCl}
\textsf{E}\{{Q}_{q}\}=\frac{r^{2}\textsf{E}\left \{ S_T^{2} \right \}}{2(1-\delta)},
\end{IEEEeqnarray}
where ${Q}_{q}$ and $\delta={r}/{\mu_i}$ are the queue length and the traffic intensity respectively, and $\textsf{E}\left \{ S_T^{2} \right \}$ is the second moment of the service distribution. Using~\eqref{st}, we can get
$$\textsf{E}\left \{ S_T^{2} \right \}=\frac{2t^{2}_p}{p^{2}_{s}}-\frac{t^{2}_p}{p_{s}}.$$
%In order to keep the queue steady, we must have $\delta<1$, or equivalently, $r<\mu$. Otherwise, if $r>\mu$, the packets mean arrival rate is greater than the mean service rate. So the queue size increases without bound over time. It is not intuitive, however, to explain why no steady-state solution exits when $r=\mu$. We can explain that as the queue grows, it is more and more difficult for the server to decrease the queue because the average service rate is no higher than the average arrival rate.

It is known that for an M/G/1 queue the average waiting time of a packet is composed of queuing and service time, and the queuing delay is $\textsf{E}\{D_q\}=\frac{\textsf{E}\{{Q}_{q}\}}{r}$. In summary, the whole delay for transmitting a packet is given by
\begin{IEEEeqnarray}{rCl}
\textsf{E}\{D\}=\frac{\textsf{E}\{{Q}_{q}\}}{r}+\textsf{E}\left\{S_T\right\}=\frac{2bR_{s}L-rL^{2}}{2b^{2}R^{2}_{s}p_{s}-2rbR_{s}L}.\label{delay}
\end{IEEEeqnarray}
We can see that the delay of serving a packet is closely related to SER, i.e., delay tolerant service has large SER, while delay sensitive service has small SER. Thus, when $\rho=1$, all the users should have small SER. Then the transmit power of MU-MIMO is in dominant place compared to the circuit power. When $\rho=0$, all the users can have large SER, where the transmit power of MU-MIMO can be very small.

\subsection{Adaptive MU-MIMO/SIMO Transmission}
For different types of services, the packets from which can have different delay profiles. For example, for voice and video services, the packets should be received in a strict delay. For mail services, packets can be transmitted with a long delay. Thus, we can divide different packets into two different delay tolerance levels.
In order to realistically analyze the energy efficiency of wireless networks, it is essential to identify the fraction of subscribers from the entire population based on the delay demands. As a consequence, the energy efficient MU-MIMO systems will be derived with partial users having delay sensitive traffic, and the others have delay tolerant traffic, that is, $\rho$ need to be considered.

Since the total circuit power of MU-SIMO is less than MU-MIMO, MU-SIMO can be more energy efficient than MU-MIMO for some $\rho$. For given throughput, in order to get the energy efficient transmission for different $\rho$, we need to consider antenna selection for MU-MIMO to minimize the total power consumption.
%Since the total circuit power $P_c$ for MU-SIMO and MU-MIMO are different, the transmission power $P$ and the circuit power $P_{c}$ can alternately become the dominant part of the total power consumption at particular delay performance for MU-SIMO and MU-MIMO.
%In other words, the SIMO systems can be superior to the MIMO systems in terms of the energy efficiency for some delay constraint, since there are less transmission antennas in the SIMO systems.
Therefore, for different ratios of delay tolerant users to delay sensitive users, we will study the switching strategy between the MU-SIMO and MU-MIMO systems by the energy efficiency criteria.

%We assume that the user have the same transmission rate for the SIMO transmission mode and the MIMO transmission mode.
%To have realistic but also analytical transmission strategy, we will study the relationship between MIMO and SIMO with respect to energy-efficiency.
We consider the MU-SIMO systems by performing antenna selection from the user's antennas in the MU-MIMO systems. Let $h^i_{kj}$ be the channel fading coefficient from the $j$-th transmit antenna to the $k$-th receive antenna for user $i$. Then the best channel gain for user $i$ is chosen as
\begin{equation}
\begin{aligned}
g^i_{SIMO}=\max_{j \in \left\{1,\cdots,N_{t}\right\}}\mathbf{h}^{H}_j\mathbf{h}_j=\max_{j \in \left\{1,\cdots,N_{t}\right\}}\sum_{k=1}^{N_{r}}\left |h^i_{kj} \right |^{2},\label{gsimo}
\end{aligned}
\end{equation}
where $\mathbf{h}_j$ is the $j$-th column vector of $\mathbf{H}_i$, $i=1,\cdots,K$.

By MMSE receiver and based on~\eqref{SINR}, we can get the SINR of the user $i$ as~\eqref{simo-sinr},
\begin{figure*}[!t]
\normalsize
\begin{equation}\label{simo-sinr}
SINR^i_{s}=\frac{P_{s}^{i}(1-\tau^2)\left(f^{s}_1(\mathbf{\lambda})+f^{s}_2(\mathbf{\lambda})\right)}{(P-P^{i}_{s})(1-\tau^2)(f^{s}_2(\mathbf{\lambda})-Kf^{s}_1(\mathbf{\lambda}))/(1-K)+K(K+1)P^{i}_{s}\tau^2f^{s}_3(\lambda)+K(K+1)\sigma^2f_3(\mathbf{\lambda})},
\end{equation}
\begin{equation}
\label{simo-p}
P^{i}_{s}=\frac{\eta(b)\left [ P(1-\tau^2)(f^{s}_2(\mathbf{\lambda})-Kf_1(\mathbf{\lambda}))/(1-K)+f_3(\mathbf{\lambda})K(K+1)\sigma^2 \right ]}{\left ( f^{s}_1(\lambda)+f^{s}_2(\lambda) \right )(1-\tau^2)+\eta(b)\left [(1-\tau^2)(f^{s}_2(\mathbf{\lambda})-Kf^{s}_1(\mathbf{\lambda}))/(1-K)-f^{s}_3(\lambda)\tau^2K(K+1) \right ]}.
\end{equation}
\hrulefill
\vspace*{10pt}
\end{figure*}
where $P_{s}^{i}$ is the transmission power for user $i$, and
\begin{IEEEeqnarray}{rcl}
\begin{aligned}
f^{s}_1(\mathbf{\lambda})&\triangleq& \sum_{l=1}^{K}\left ( \frac{\lambda_l}{\lambda_l+\alpha} \right )^2,\\
f^{s}_2(\mathbf{\lambda})&\triangleq& \left ( \sum_{l=1}^{K}\frac{\lambda_l}{\lambda_l+\alpha} \right)^2,\\
f^{s}_3(\mathbf{\lambda})&\triangleq& \sum_{l=1}^{K}\frac{\lambda_l}{(\lambda_l+\alpha)^2}.
\end{aligned}
\end{IEEEeqnarray}
From~\eqref{gammas} and~\eqref{pe}, the transmission power for the user $i$ is shown in~\eqref{simo-p},
and the total power consumption for MU-SIMO systems is
\begin{IEEEeqnarray}{rCl}
 \hat{P}^{opt}_s=\sum_{i=1}^{K}\left(P^{i}_{s}+ P_{0}\right).
\end{IEEEeqnarray}

%Since we consider the SIMO and MIMO with the same transmission rate, the throughput $T=\frac{L_{p}bR_{s}P_{s}}{L^{opt}}$ is the same for given successful probability $P_{s}$.
To select the transmission mode with the maximum energy efficiency, we only need to select the transmission mode which consumes less power at the same throughput, that can be denoted as
\begin{IEEEeqnarray}{rCl}
t^{\ast}=\arg\min_{t \in \left\{m,s\right\}}\hat{P}^{opt}_t,
\end{IEEEeqnarray}
%where $f_{ee}(m)=\frac{T}{P(m)}$, $f_{ee}(s)=\frac{T}{P(s)}$ represent the energy-efficient functions of the optimal MIMO, and SIMO transmission, respectively.
where $m$ and $s$ stand for MU-MIMO and MU-SIMO modes respectively. Therefore, we can adaptively change the transmission mode to meet different users' QoS of delay, and allocate the transmission power to ensure the optimal energy efficiency of the MU-MIMO systems at the same time. This is summarized in the following theorem.

\newtheorem{theorem}{Theorem}
\begin{theorem}\label{t2}
There exists a ratio $\rho^{\ast}$, i.e., the number of the heavy users and the other users are $\rho^{\ast}K$ and $(1-\rho^{\ast})K$ respectively, such that all the users should use MIMO mode when $\rho>\rho^{\ast}$, and  all the users should use SIMO mode when $\rho<\rho^{\ast}$. Then, we can realize the energy efficient transmission. Denote $p_{ij}(p^1_{e})$ and $p_{ij}(p^2_{e})$ as the $j$-th stream's transmission power for the delay sensitive user $i$ and the delay tolerant user $i$ respectively. Then $\rho^{\ast}$ is the solution of following equation
\begin{equation}\nonumber
\begin{aligned} \sum_{i=1}^{\rho^{\ast}K}\sum_{j=1}^{N_t}p_{ij}(p^1_{e})+\sum_{i=1}^{(1-\rho^{\ast})K}\sum_{j=1}^{N_t}p_{ij}(p^2_{e})+KN_tP_0=\\
\sum_{i=1}^{\rho^{\ast}K}P^i_s(p^1_{e})+\sum_{i=1}^{(1-\rho^{\ast})K}P^i_s(p^2_{e})+KP_0,
\end{aligned}
\end{equation}
where $P^i_s(p^1_{e})$ and $P^i_s(p^2_{e})$ are the transmission power for the SIMO mode with two different delay-aware services.
\end{theorem}
\begin{proof}
The proof of Theorem~\ref{t2} is given in Appendix~\ref{a3}.
\end{proof}

\section{Simulation Results and Discussion}
This section presents simulation results to evaluate the theoretical results of MU-MIMO delay-aware energy efficient communications with the joint consideration of the heterogeneous traffic delays and the mode switching.

\subsection{System Parameters}
Unless specified otherwise, for all simulations, we assume that the packet size $L=1080$, %the maximum retransmission times of packet $N^{max}_{r}=6$,
the size of the header bits $L_h=32$, and the average arrival rate of packets, $r=1$ packets/time-unit. The symbol rate $R_s=100$KHz, the number of receiving antennas $N_r=4$, the total number of users $K=10$, and each user has two antennas. %and the noise variance $\sigma^2=-50$dB/Hz.
For the purpose of simplicity, we assume that the SER $p^1_{e}=10^{-2}$ is for delay tolerant service, and the SER $p^2_{e}=10^{-4}$ is for delay sensitive service. Our simulation results are averaged over 1000 channel realizations.%Like 3GPP LTE standard, we assume that block~(also called frame) length $T_u=2$ ms.

\subsection{Performance Evaluation}
In order to validate the impact of using $\hat{\mathbf{W}}$ to simplify $\hat{\mathbf{W}}_{MMSE}$ is small and the analytical results, we simulate and plot the figure of received SINR for the stream $1$ of user $1$. Fig.~\ref{Fig-influence} shows the results of the change. It can be observed that the variations of the received SINR agree reasonably well. From this figure we can also observe that the received SINR of $\hat{\mathbf{W}}_{MMSE}$ is larger than that of $\hat{\mathbf{W}}$, this comes from the fact that the simplified receiver can cause performance loss.
\begin{figure}[!t]
\begin{center}
\includegraphics [width=3.8in]{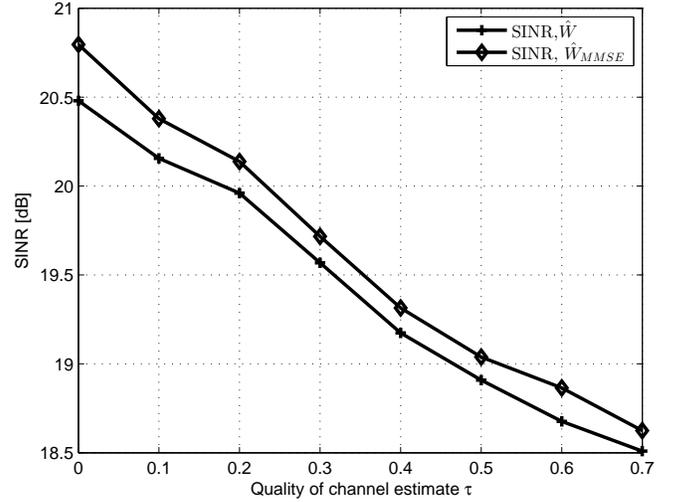}
\caption{The received SINR for stream $1$ of user $1$ with using $\hat{\mathbf{W}}$ to simplify $\hat{\mathbf{W}}_{MMSE}$.
 } \label{Fig-influence}
\end{center}
\end{figure}

Fig.~\ref{Fig-transmit} shows the total transmit power of the MU-MIMO systems versus the quality of channel estimation. We can see from the figure that the total transmit power is increasing when increasing $\tau$. Specifically, we consider two circumstances: all the users have delay sensitive traffic, i.e., $\rho=1$, and all the users have delay tolerant traffic, i.e., $\rho=0$. We can see that when we use the energy efficient modulation for the MU-MIMO systems, the total transmit power can be always smaller than that by allocating equal transmit rate for each spatial stream.
\begin{figure}[!t]
\begin{center}
\includegraphics [width=3.8in]{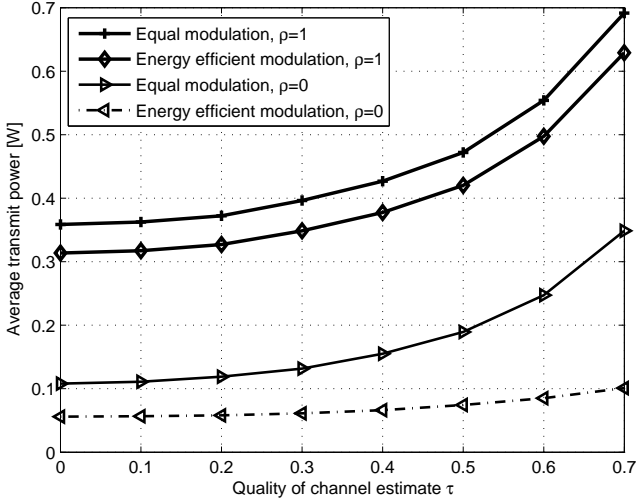}
\caption{Total transmit power of the MU-MIMO systems versus the quality of the channel estimate $\tau$.
 } \label{Fig-transmit}
\end{center}
\end{figure}

On the other hand, the total transmit power of the MU-MIMO systems under the condition of $\rho=1$ is much larger than that of $\rho=0$. Since all the users have delay sensitive traffic for $\rho=1$, the user's transmit power should be large to guarantee the delay requirement, causing the total transmit power large. The characteristic of Fig.~\ref{Fig-transmit} validates the theoretical derivation. Therefore, we can use the energy efficient modulation for the spatial streams of MU-MIMO systems, and the energy efficient modulation can be variable according to different delay requirements of user's service.

Fig.~\ref{Fig-ee} shows the energy efficiency of the MU-MIMO systems versus the channel estimation quality $\tau$. The SIMO mode means that all the users use the SIMO mode based on antenna selection from~\eqref{gsimo}. We observe from the figure that the energy efficiency of MU-SIMO is larger than that of MU-MIMO in the regime of small $\tau$, due to more circuit power and small transmit power. On the other hand, in the regime of large $\tau$, the energy efficiency of MU-MIMO is larger than that of MU-SIMO, which comes from the fact that more transmit power is consumed when increasing $\tau$. Then the circuit power is negligible. At the same time, the energy efficiency of both MU-MIMO and MU-SIMO are decreasing when increasing $\tau$. The reason is that the transmit power would be larger when the channel estimation quality is worse. However, the throughput is the same value. Consequently, the energy efficiency of the MU-MIMO and MU-SIMO are decreasing with larger $\tau$.
\begin{figure}[!t]
\begin{center}
\includegraphics [width=3.8in]{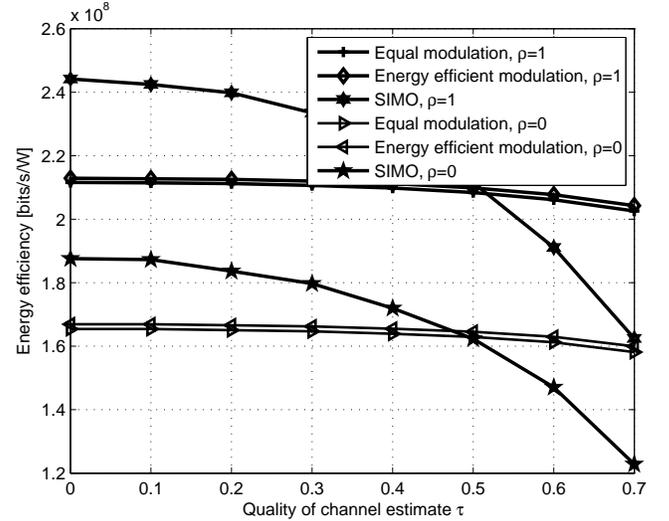}
\caption{Energy efficiency of the MU-MIMO systems versus the quality of the channel estimate: $\tau$.
 } \label{Fig-ee}
\end{center}
\end{figure}

In order to have a better explanation of switching between MU-MIMO and MU-SIMO, the effect of the ratio $\rho$ to the energy efficiency of different transmission modes stated in Theorem~\ref{t2} is plotted and compared in Fig.~\ref{Fig-ratio}. We can see that the energy efficiency is decreasing when increasing $\rho$ for MU-MIMO mode and MU-SIMO mode, which indicates more power consumption at the same throughput when increasing $\rho$. Furthermore, Fig.~\ref{Fig-ratio} shows that there exits a crossover point $\rho^{\ast}$ between the MU-MIMO and the MU-SIMO transmission, which is consistent with the analytic results. When the ratio $\rho$ is smaller than that corresponding to the crossover point, i.e., $\rho<\rho^{\ast}$, the MU-SIMO is superior in energy efficiency. Otherwise, the MU-MIMO offers better energy efficiency. These results further indicate that the energy efficiency can be improved by turning off the antennas with low gain if the total number of users having delay-sensitive services is very small. Therefore, we can choose the better energy efficient transmission mode between MU-MIMO and MU-SIMO for the multi-user systems according to the ratio of users who have the delay-sensitive services.
\begin{figure}[!t]
\begin{center}
\includegraphics [width=3.8in]{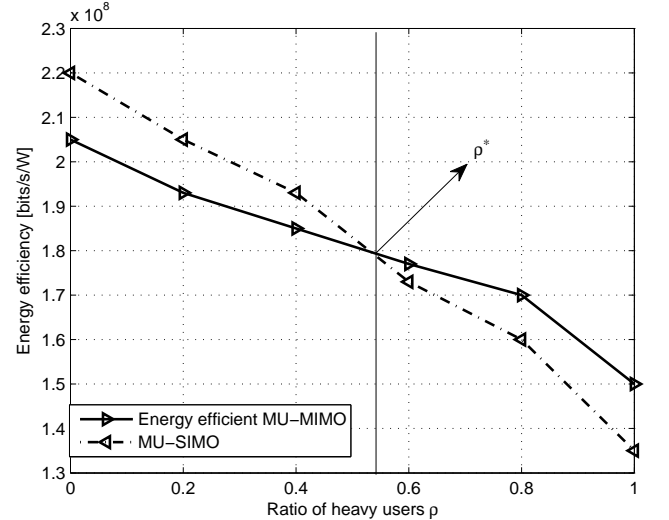}
\caption{Energy efficiency of the MU-MIMO systems and MU-SIMO versus the ratio of heavy users: $\rho$.
 } \label{Fig-ratio}
\end{center}
\end{figure}

\section{Conclusion }
In this paper, we first investigate the energy efficient modulation for the spatial streams of the MU-MIMO systems, which is based on MMSE receiver. Then, we study the heterogeneous delay-aware energy efficiency with different ratios of heavy users. To achieve energy efficient communications for MU-MIMO, we design the energy efficient constellation size allocation for the spatial streams of each user. By considering the dominance of circuit power consumption and transmit power consumption at different conditions, we find the crossover point and propose a switching strategy to select the energy efficient transmission mode between MU-MIMO and MU-SIMO. The strategy is based on the ratio of the number of the delay sensitive users to that of the delay tolerant users, which can guarantee the energy efficiency of the delay-aware MU-MIMO systems.

\appendices
\section{proof of Lemma~\ref{p1}}\label{a1}
From~\cite{2004-random-matrix}, an $n\times n$ random matrix $\mathbf{Q}=[q_{i,j}]^{n}_{i,j=1}$ is unitarily invariant if it is uniformly distributed on the set, $\mathcal{Q}(n)$, of the $n\times n$ unitary matrices.

For the probability measure $\gamma_n$ on $\mathcal{Q}(n)$, we consider the probability space $(\mathcal{Q}(n), \gamma_n)$ and have the expectation $E\left(f\right):=\int f(\mathbf{Q})d\gamma_n(\mathbf{Q})$ for a measure function $f$.
The invariance of $\gamma_n$ guarantees that $E\left(f(\mathbf{Q})\right)=E\left(f(\mathbf{V}\mathbf{Q})\right)$ is valid
for any $\mathbf{V}\in \mathcal{Q}(n)$. When $\mathbf{V}=\mbox{diag}(e^{i\theta_1},\cdots, e^{i\theta_n})$, we have
\begin{equation}
E\left(f\right)=E\left(f\left([e^{i\theta_i}q_{i,j}]^{n}_{i,j=1}\right)\right)\label{in_unitary}
\end{equation}
for all $\theta_i\in \mathcal{R}$.
From~\eqref{in_unitary}, we have
$$E\left(\left ( \begin{bmatrix}
         \mbox{cos}\theta & \mbox{sin}\theta\\
         -\mbox{sin}\theta & \mbox{cos}\theta
         \end{bmatrix}\bigoplus \mathbf{I}_{n-2} \right )\mathbf{Q}\right)=E\left(\mathbf{Q}\right).$$
The random variables $q_{i,j}$ are identically distributed. Thus we have $q^{'}_{1,1}=q_{1,1}\mbox{cos}\theta+q_{2,1}\mbox{sin}\theta$ and
$q^{'}_{2,1}=q_{2,1}\mbox{cos}\theta-q_{1,1}\mbox{sin}\theta$. By~\eqref{in_unitary}, we have
\begin{equation}\label{q1=q2} E\left ( \left | q_{1,1} \right |^2\left|q_{2,2}\right|^2 \right )=E\left ( \left | q^{'}_{1,1} \right |^2\left|q^{'}_{2,2}\right|^2 \right ). \end{equation}
We can get
\begin{equation}\label{q1} \begin{aligned}
\left | q^{'}_{1,1} \right |^2=\left | q_{1,1} \right |^2\mbox{cos}^2\theta +\left | q_{2,1} \right|^2\mbox{sin}^2\theta    \\
+(q_{1,1}\bar{q}_{2,1}+\bar{q}_{1,1}q_{2,1})\mbox{cos}\theta\mbox{sin}\theta,
\end{aligned}
\end{equation}
and
\begin{equation}\label{q2}
\begin{aligned}
\left | q^{'}_{2,2} \right |^2=\left |q_{2,2} \right |^2\mbox{cos}^2\theta+\left | q_{1,2} \right|^2\mbox{sin}^2\theta   \\
-(q_{2,2}\bar{q}_{1,2}+\bar{q}_{2,2}q_{1,2})\mbox{cos}\theta\mbox{sin}\theta.
\end{aligned}
\end{equation}
Substituting~\eqref{q1},~\eqref{q2} into~\eqref{q1=q2} we can get
\begin{equation}\label{q3}\small
E\left((q_{1,1}\bar{q}_{2,1}+\bar{q}_{1,1}q_{2,1})(q_{2,2}\bar{q}_{1,2}+\bar{q}_{2,2}q_{1,2})\right)=\frac{-2}{N_tK(N^2_tK^2-1)}.
\end{equation}
We know that~(see~\cite{2004-random-matrix})
\begin{equation}
\begin{aligned}\label{q4}
E\left(q_{1,1}q_{2,2}\bar{q}_{2,1}\bar{q}_{1,2}\right)&=E\left(q_{2,1}q_{1,2}\bar{q}_{1,1}\bar{q}_{2,2}\right)\\
&=-\frac{1}{N_tK(N^2_tK^2-1)}.
\end{aligned}
\end{equation}
Therefore, apply~\eqref{q4} to~\eqref{q3} to get $E\left(q_{ij}q_{ij^{'}}\bar{q}_{i^{'}j}\bar{q}_{i^{'}j^{'}} \right )=\frac{-1}{N_tK(N^2_tK^2-1)}$.

\section{proof of Proposition~\ref{p2}}\label{a2}
The transmission power~\eqref{pib} becomes
\begin{equation}
\begin{aligned}
p_{ij}=\frac{f(b_{ij})\left [ c_3P+f_3(\mathbf{\lambda})\hat{\sigma}^2 \right ]}{c_1+c_2\eta(b_{ij})} \\
=\frac{c_3P}{\frac{c_1}{\eta(b_{ij})}+c_2}+\frac{f_3(\mathbf{\lambda})\hat{\sigma}^2}{\frac{c_1}{\eta(b_{ij})}+c_2}.
\end{aligned}
\end{equation}
For $\sum_{i=1}^{K}\sum_{j=1}^{N_t}p_{ij}=P$, we obtain
$$P=c_3P\sum_{i=1}^{K}\sum_{j=1}^{N_t}\frac{1}{\frac{c_1}{\eta(b_{ij})}+c_2}+\sum_{i=1}^{K}\sum_{j=1}^{N_t}\frac{f_3(\mathbf{\lambda})\hat{\sigma}^2}{\frac{c_1}{\eta(b_{ij})}+c_2},$$
and therefore
$$P=\frac{\sum_{i=1}^{K}\sum_{j=1}^{N_t}\frac{\hat{\sigma}^2f_3(\mathbf{\lambda})}{c_2+\frac{c_1}{\eta(b_{ij})}}}{1-c_3\sum_{i=1}^{K}\sum_{j=1}^{N_t}\frac{1}{c_2+\frac{c_1}{\eta(b_{ij})}}}.$$
%For the SIMO mode, all the users should use only one antenna, i.e., $N_t=1$, we only need to substitute $b$ into $f(b_i)$.

\section{proof of Theorem~\ref{t2}}\label{a3}
On the condition of $\rho=1$, all the users have delay sensitive traffic. Then $p_e$ should be small. From~\eqref{SINRb} and \eqref{pib}, when the SER $p_{e}$ is small, we know that the transmission power $p_i$ will be large and dominates the total power consumption, since $p_i\rightarrow \infty$ as $p_e\to 0$. Hence the circuit power is negligible compared to the transmission power and $\hat{P}_m\approx P$, $\hat{P}^{opt}_s\approx \sum_{i=1}^{K}P^{i}_{s}$, and the transmission SNR is high.
By~\eqref{gammas} and~\eqref{pe}, for a particular $p_e$ and the same symbol transmission, we assume that the MIMO systems transmit the same copy of the symbol $s_s$ per antenna. Then we have
$$\gamma^s_{s}=\gamma^m_{s}=\text{SINR}\frac{B}{R_{s}}.$$

For each user $i$, the received SINR of the SIMO, $SINR_s$, and the received SINR of MIMO, $SINR_m$, have the relation of $SINR_s=SINR_m$. From~\eqref{receiveds} and based on zero-forcing~(ZF) receiver, we can get
$$\hat{\mathbf{s}}=\sqrt{P^{i}_m}\mathbf{\Lambda}\tilde{\mathbf{s}}+\mathbf{U}^{H}\mathbf{n},$$ where $\tilde{\mathbf{s}}=\underbrace{\left [ s_s,s_s,\cdots,s_{s} \right ]^{T}}_{N_t}$ and $P^{i}_{m}$ is the transmission power for the transmission symbol $s_s$ with constellation size $b=\sum_{r=1}^{N_{t}}b_r$ of user $i$.  Therefore $$SNR_m=\frac{P^{i}_m\mbox{Tr}\left\{\mathbf{\Lambda}\mathbf{\Lambda}^{H}\right\}}{N_0B}=\frac{P^{i}_{m}\left \|\mathbf{H}_i\right \|^{2}_F}{N_0B}.$$
This is equivalent to the SNR of the symbol transmission with the space-time block coding\cite{2009-Huang-SE}.

For the channel $\mathbf{H}_i$, when we select the transmit antenna with best channel gain to the receive antennas from~\eqref{gsimo}, we can get
$$SNR_s=\frac{P^{i}_{s}g^{i}_{SIMO}}{N_0B}.$$
Therefore, we have $P^i_m=\frac{SNR_mN_oB}{\left \|\mathbf{H}_i\right \|^{2}_F}$, and $P^{i}_{s}=\frac{SNR_sN_oB}{g^{i}_{SIMO}}$.
Note that
\begin{equation}\nonumber
\begin{aligned}
g^i_{SIMO}=\max_{j \in \left\{1,\cdots,N_{t}\right\}}\sum_{r=1}^{N_{r}}\left |h^i_{rj} \right |^{2},
\end{aligned}
\end{equation}
and
$$\left \| \mathbf{H}_i \right \|_F^{2}=\sum_{j=1}^{N_{t}}\lambda_{(i-1)K+j}^{2}=\sum_{r=1}^{N_{r}}\sum_{j=1}^{N_t}\left |h^i_{rj} \right |^{2}.$$
We have $\left \| \mathbf{H}_i \right \|_F^{2}> g^i_{SIMO}$, which results in $P^i_m<P^{i}_{s}$ with $SNR_s=SNR_m$. In the high SNR regime, in comparison with multiplexing and space time block coding, we can get $P^i_m(\mathbf{b})\leq P^i_m$ at the same transmission rate~\cite{2009-Huang-SE}. Then we have $P^i_m(\mathbf{b})<P^{i}_{s}$. Therefore
$$f^s_{ee}\approx T/\sum_{i=1}^K P^{i}_{s}<T/\sum_{i=1}^K P^i_m(\mathbf{b})\approx f^m_{ee}.$$
%thus, from Eq.~\eqref{transm} and Eq.~\eqref{transs} we can get $P_t(m)<P_t(s)$, making $f_{ee}(s)<f_{ee}(m)$.
This shows that MIMO mode outperforms the SIMO mode in terms of energy efficiency when all the users have the fixed rate $b$, that is the optimal transmission mode $t^{\ast}$=$m$, where $m$ stands for the MIMO mode.

On the other hand, when $\rho=0$, all the users have delay tolerant services. Then $p_e$ can be large. When $p_{e}$ is close to $1$, the circuit power $P_{0}$ will dominate the total power consumption, that is $\hat{P}_m\approx P^m_c$ and $\hat{P}^{opt}_s\approx P^s_c$. Since
$$P^m_c=\sum_{i=1}^{K}\sum_{j=1}^{N_t}P_0 >\sum_{i=1}^{K}P_0=P^s_c,$$
we have
$$f^s_{ee}\approx T/P^s_c>T/P^m_c\approx f^m_{ee}.$$
This shows that the SIMO transmission mode can be selected to improve the energy efficiency, resulting in the optimal transmission mode $t^{\ast}$=$s$, where $s$ stands for SIMO mode.

%By~\eqref{delay}, the average delay $\bar D$ is a function of the successful transmission probability $p_s$. Since $p_{s}=(1-p_{e})^{N_tL/b}$, therefore, $\bar D$ is a function of $p_e$. Since $f^s_{ee}<f^m_{ee}$ for small $p_e$ and $f^s_{ee}>f^m_{ee}$ for large $p_e$, there must exist a $p_{e}^{\star}$ such that  $f^s_{ee}=f^m_{ee}$. This means that
%$$\bar D^{\star}=\frac{2bR_{s}L-rL^{2}}{2b^{2}R^{2}_{s}p^{\star}_{s}-2rbR_{s}L}$$
%is the delay threshold value for selecting the transmission mode, where $p^{\star}_{s}=(1-p_{e}^{\star})^{N_tL/b}$.
In all, there exists a ratio $\rho^{\ast}$ making the energy efficiency of the MU-MIMO equal to that of MU-SIMO. The number of delay-sensitive users is $\rho^{\ast}K$, and the number of delay-tolerant services is $(1-\rho^{\ast})K$. For MU-MIMO mode, the total power consumption $\hat{P}_m=\sum_{i=1}^{\rho^{\ast}K}\sum_{j=1}^{N_t}p_{ij}(p^1_{e})+\sum_{i=1}^{(1-\rho^{\ast})K}\sum_{j=1}^{N_t}p_{ij}(p^2_{e})+KN_tP_0$. For the MU-SIMO mode, the total power consumption is $\hat{P}^{opt}_s=\sum_{i=1}^{\rho^{\ast}K}P^i_s(p^1_{e})+\sum_{i=1}^{(1-\rho^{\ast})K}P^i_s(p^2_{e})+KP_0$. Thus, $\rho^{\ast}$ is the solution of the following equation
\begin{equation}\nonumber
\begin{aligned} \sum_{i=1}^{\rho^{\ast}K}\sum_{j=1}^{N_t}p_{ij}(p^1_{e})+\sum_{i=1}^{(1-\rho^{\ast})K}\sum_{j=1}^{N_t}p_{ij}(p^2_{e})+KN_tP_0=\\
\sum_{i=1}^{\rho^{\ast}K}P^i_s(p^1_{e})+\sum_{i=1}^{(1-\rho^{\ast})K}P^i_s(p^2_{e})+KP_0.
\end{aligned}
\end{equation}

%\newcounter{mytempeqncnt}
%\begin{figure*}[!t]
%\normalsize
%\setcounter{mytempeqncnt}{\value{equation}}
%\setcounter{equation}{5}
%\begin{equation}
%\label{eqn_dbl_x}
%x = 5 + 7 + 9 + 11 + 13 + 15 + 17 + 19 + 21+ 23 + 25
%+ 27 + 29 + 31
%\end{equation}
%\setcounter{equation}{\value{mytempeqncnt}}
%\hrulefill
%\vspace*{4pt}
%\end{figure*}
%

%% use section* for acknowledgement
%\section*{Acknowledgment}
%This work is supported by the National 973 Project \#2012CB316106, and
%by NSF China  \#61161130529.
%This work is supported by  NSF China \#60972031, by SEU SKL
%project \#W200907 , by ISN project \#ISN11-01, by National 973 project \#2009CB824900.




\bibliography{mybib}

\begin{thebibliography}{10}

\bibitem{2011-Li-EE}
G.~Y. Li, Z.~Xu, C.~Xiong, C.~Yang, S.~Zhang, Y.~Chen, and S.~Xu,
  ``Energy-efficient wireless communications: tutorial, survey, and open
  issues,'' {\em IEEE Wireless Commun. Mag.}, vol.~18, no.6, pp.~28--35, Dec
  2011.

\bibitem{2002-Paulraj-MIMO}
A.~J. Paulraj, D.~A. Gore, R.~U. Nabar, and H.~Bolcskel, ``An overview of mimo
  communicatios-a key to gigabit wireless,'' in {\em Proceedings of the IEEE.},
  vol.~92, no.2, pp.~198--218, Feb. 2002.

\bibitem{2009-Soysal-MIMO}
A.~Soysal and S.~Ulukus, ``Optimality of beamforming in fading mimo multiple
  access channels,'' {\em IEEE Trans. Commun.}, vol.~57, no.4, pp.~1171--1183,
  Apr. 2009.

\bibitem{2007-Gesbert-MIMO}
D.~Gesbert, M.~Kountouris, J.~R.~W.~Heath, C.~B. Chae, and T.~Salzer,
  ``Shifting the mimo paradigm,'' {\em IEEE Signal Process. Mag.}, vol.~24,
  no.5, pp.~36--46, Sep. 2007.

\bibitem{2005-S-Cui-Modulation}
S.~Cui, A.~J. Goldsmith, and A.~Bahai, ``Energy-constrained modulation
  optimization,'' {\em IEEE Trans. Wireless Commun.}, vol.~5, pp.~2349--2360,
  September 2005.

\bibitem{2011-Costa-EE}
F.~M. Costa and H.~Ochiai, ``Energy-efficient physical layer design for
  wireless sensor network links,'' in {\em Proc. of the IEEE ICC}, (Kyoto,
  Japan), pp.~1--5, Jun. 2011.

\bibitem{2014-Wu-EE}
J.~Wu, G.~Wang, and Y.~R. Zhang, ``Energy efficiency and spectral efficiency in
  type-\uppercase\expandafter{\romannumeral1} arq systems,'' {\em IEEE J. Sel.
  Areas Commun.}, vol.~32, no.4, pp.~356--366, Feb. 2014.

\bibitem{2004-S-Cui-MIMO}
S.~Cui, A.~J. Goldsmith, and A.~Bahai, ``Energy-efficiency of mimo and
  cooperative mimo techniques in sensor networks,'' {\em IEEE Jour. Select.
  Aeras in Commun}, vol.~22, pp.~1089--1098, Aug 2004.

\bibitem{2013-Miao-MIMO}
G.~Miao, ``Energy-efficient uplink multi-user mimo,'' {\em IEEE Trans. Wireless
  Commun.}, vol.~12, no.5, pp.~2302--2313, May 2013.

\bibitem{2009-Kim-MIMO}
H.~Kim, C.~B. Chae, G.~de~Veciana, and R.~W. Heath, ``A cross-layer approach to
  energy efficiency for adaptive mimo systems exploiting spare capacity,'' {\em
  IEEE/ACM Trans. Networking}, vol.~8, pp.~4264--4275, August 2009.

\bibitem{2002-Gallager-Delay}
R.~A. Berry and R.~G. Gallager, ``Communication over fading channels with delay
  constraints,'' {\em IEEE Trans. Inform. Theory}, vol.~48, no.5,
  pp.~1135--1149, May 2002.

\bibitem{2004-Rajan-Delay}
D.~Rajan, A.~Sabharwal, and B.~Aazhang, ``Delay-bounded packet scheduling of
  bursty traffic over wireless channels,'' {\em IEEE Trans. Inform. Theory},
  vol.~50, no.1, pp.~125--144, Jan 2004.

\bibitem{2007-Cui-Cross-layer}
S.~Cui, A.~J.~G. R.~Madan, and S.~Lall, ``Cross-layer energy and delay
  optimization in small-scale sensor networks,'' {\em IEEE Trans. Wireless
  Commun.}, vol.~6, pp.~3688--3699, Oct 2007.

\bibitem{2012-Leila-EE}
L.~Musavian and T.~Le-Ngoc, ``Energy-efficient power allocation for
  delay-constrained systems,'' in {\em Proc. of the IEEE GLOBECOM, Anaheim,
  CA}, pp.~3554--3559, Dec 2012.

\bibitem{2014-Lin-5G}
C.-L. I, C.~Rowell, S.~Han, Z.~Xu, G.~Li, and Z.~Pan, ``Toward green and soft:
  A 5g perspective,'' {\em IEEE Commun. Mag.}, vol.~52, pp.~66--73, Feb. 2014.

\bibitem{2007-Han-delay}
Z.~Han, X.~Liu, Z.~J. Wang, and K.~J.~R. Liu, ``Delay sensitive scheduling
  schemes for heterogeneous qos over wireless networks,'' {\em IEEE Trans.
  Wireless Commun.}, vol.~6, no.2, pp.~423--428, Feb. 2007.

\bibitem{2009-Poor-EE}
F.~Meshkati, H.~Poor, S.~Schwartz, and R.~Balan, ``Energy-efficient resource
  allocation in wireless networks with quality-of-service constraints,'' {\em
  IEEE Trans. Commun.}, vol.~57, no.11, pp.~3406--3414, Nov 2009.

\bibitem{2009-Farhad-EE}
F.~Meshkati, H.~V. Poor, and S.~C. Schwartz, ``Energy efficiency-delay
  tradeoffs in cdma networks: a game-theoretic approach,'' {\em IEEE Trans.
  Inf. Theory}, vol.~55, no.7, pp.~3220--3228, Jul. 2009.

\bibitem{2011-Andrews-MIMO}
B.~Nosrat-Makouei, J.~G. Andrews, and J.~Robert W.~Heath, ``Mimo interference
  alignment over correlated channels with imperfect csit,'' {\em IEEE Trans.
  Signal Process.}, vol.~59, no.6, pp.~2783--2794, Jun. 2011.

\bibitem{2007-Atilla-Eryilmaz-Congestion}
R.~S. Atilla~Eryilmaz, ``Fair resource allocation in wireless networks using
  queue-length-based scheduling and congestion control,'' {\em IEEE/ACM Trans.
  Networking}, vol.~15, no.6, pp.~1333--1344, Dec 2007.

\bibitem{2000-Pursley-PSK}
M.~B. Purslet and J.~M. Shea, ``Adaptive nonuniform phase-shift-key modulation
  for multimedia traffic in wireless networks,'' {\em IEEE J. Select. Aeras
  Commun}, vol.~18, pp.~1394--1407, Aug 2000.

\bibitem{2001-Lavery-thesis}
R.~J. Lavery, ``Throughput optimization for wireless data transmission,''
  Master's thesis, Polytechnic University, June 2001.

\bibitem{2002-Catreux-Throughputs}
S.~Catreux, P.~F. Driessen, and L.~J. Greenstein, ``Data throughputs using
  multiple-input multiple-output (mimo) techniques in a noise-limited cellular
  environment,'' {\em IEEE Trans. Wireless Commun}, vol.~1, pp.~226--235, Apr
  2002.

\bibitem{2015-wkl-MIMO}
K.~Wang and W.~Chen, ``Energy-efficient communications in mimo systems based on
  adaptive packets and congestion control with delay constraints,'' {\em IEEE
  Trans. Wireless Commun.}, vol.~14, no.4, pp.~2169--2179, Apr. 2015.

\bibitem{1998-Verdu-Multiuser}
S.~Verdu, {\em Multiuser Detection}.
\newblock Cambridge University Press, 1998.

\bibitem{2005-RZF}
C.~B. Peel, B.~M. Hochwald, and A.~L. Swindlehurst, ``A vector-perturbation
  technique for near-capacity multiantenna multiuser communication-part
  \uppercase{I}: Channel inversion and regularization,'' {\em IEEE Trans.
  Commun.}, vol.~53, no.1, pp.~195--202, Jan. 2005.

\bibitem{Edelman-1989-matrices}
A.~Edelman, {\em Eigenvalues and condition numbers of random matrices}.
\newblock PhD thesis, Dept. Math., Mass. Inst. Technol., Cambridge, MA, USA,
  1989.

\bibitem{2011-random-matrix}
R.~Couillet and M.~Debbah, {\em Random Matrix Methods and Wireless
  Communications}.
\newblock Cambridge, U.K.: Cambridge Univ. Press, 2011.

\bibitem{2005-Goldsmith-Wireless}
A.~J. Goldsmith, {\em Wireless Communications}.
\newblock New York, NY: Cambridge University Press, 2005.

\bibitem{2000-Roubi}
A.~Roubi, ``Method of centers for generalized fractional programming,'' {\em
  Jounal of Optimization Theory and Applications}, vol.~107, no.~1,
  pp.~123--143, 2000.

\bibitem{2004-Boyd-CO}
S.~Boyd and L.~Vandenberghe, {\em Convex Optimization}.
\newblock Cambridge University Press, 2004.

\bibitem{1985-Gross-Queueing}
D.~Gross and C.~M. Harris, {\em Fundamentals of Queueing Theory}.
\newblock New York: John Wiley \& Sons, 1985.

\bibitem{2004-random-matrix}
A.~M. Tulino and S.~Verdu, {\em Random Matrix Theory and Wireless
  Communications}.
\newblock Delft, The Netherlands: Now Publishers Inc, 2004.

\bibitem{2009-Huang-SE}
J.~Huang and S.~Signell, ``On spectral efficiency of low-complexity adaptive
  mimo systems in rayleigh fading channel,'' {\em IEEE Trans. Wireless
  Commun.}, vol.~8, no.9, pp.~4369--4374, Sep. 2009.

\end{thebibliography}
\bibliographystyle{ieeetr}

\end{document}